\documentclass[english,fleqn]{llncs}
\input{preamble}

\usepackage{todonotes}
\newcommand{\todoall}[1]{\todo[color=gray,size={\fontsize{8pt}{9pt}\selectfont}]{#1}}
\newcommand{\todoalexander}[1]{\todo[color=green,size={\fontsize{8pt}{9pt}\selectfont}]{\textbf{AK}: #1}}

\renewcommand{\todo}[2][]{}  
\title{An Institution for Simple UML State Machines}
\author{Alexander Knapp\inst{1} 
   \and Till Mossakowski\inst{2} 
   \and Markus Roggenbach\inst{3}
   \and Martin Glauer\inst{2}} 
\institute{Universit\"at Augsburg, Germany 
      \and Otto-von-Guericke Universit\"at Magdeburg, Germany
      \and Swansea University, UK}

\begin{document}

\maketitle
\begin{abstract}
We present an institution for UML state machines without hierarchical
states. The interaction with UML class diagrams is handled via
institutions for guards and actions, which provide dynamic components
of states (such as valuations of attributes) but abstract away from
details of class diagrams. We also study a notion of interleaving
product, which captures the interaction of several state machines.
The interleaving product construction is the basis for a semantics of
composite structure diagrams, which can be used to specify the
interaction of state machines.  This work is part of a larger effort
to build a framework for formal software development with UML, based
on a heterogeneous approach using institutions.

\medskip\noindent
\emph{Keywords}: UML, state machines, interleaving product, institutions
\end{abstract}

\section{Introduction}

The ``Unified Modeling Language''
(UML~\cite{uml-2.4.1-superstructure}) is a heterogeneous language: UML
comprises a language family of 14 types of diagrams of structural and
behavioural nature.  These sub-languages are linked through a common
meta-model, i.e., through abstract syntax; their semantics, however,
is informally described mainly in isolation.
In~\cite{knapp-mossakowski-roggenbach:corr:2014}, we have outlined our
research programme of ``institutionalising UML''. Our objective is to
give, based on the theory of institutions~\cite{GoguenBurstall92},
formal, heterogeneous semantics to UML, that --- besides providing
formal semantics for the individual sub-languages --- ultimately
allows to ask questions concerning the consistency between different
diagram types and concerning refinement and implementation in a system
development.  In this paper, we propose a new institution for UML
state machines.

\emph{Behavioural} UML state machines specify the behaviour of model
elements, like components, whereas \emph{protocol} UML state machines
express usage protocols, like the message exchange over a connector
between components.  Both variants describe dynamical system behaviour
in terms of action effects and messages, where conditions are used to
choose between different possibilities of the behaviour. We tackle the
well-known resulting problem of integrating specifications of data
(i.e., action effects and messages), logic (i.e., conditions), and
processes (i.e., state machines)
\cite{grosse-rhode:2004,roggenbach06,MossakowskiRoggenbach07,oreilly12b}
by a two-step semantics: In the first step, we define
\emph{institutions of guards and actions} that capture which guards,
actions, and messages can be used in order to define a state
machine. In general, other UML diagrams like class diagrams or OCL
constraints specify these items, i.e., define a suitable
environment. In a second step, we then define institutions for
behavioural and protocol state machines relative to given institutions
of guards and actions.  However, currently both of our institutions
are restricted to ``flat'', non-hierarchical state machines; in fact,
most of the hierarchical features can be reduced to this
format~\cite{schattkowsky-mueller:vlhcc:2005,fecher-schoenborn:fmics:2007}.
A previous UML state machine institution by D.~Calegari and
N.~Szasz~\cite{calegari-szasz:isse:2011} encoded all these features on
a single (signature) level thus reducing integration flexibility
considerably; furthermore, it only comprised behavioural state
machines and captured each state machine in isolation. By contrast, we
study interacting state machines and the refinement of state machines.
 
Our institution of behavioural state machines has the peculiarity of
being a ``programming language-like'' institution, in the sense that
each sentence essentially has one model, its canonical model.  By
contrast, our institution of protocol state machines is a ``loose
semantics'' institution where generally a sentence has many models.
For system development, we introduce an interleaving product of
several state machines in our institution, which allows us to consider
refinement for checking the correct implementation of protocols and
which ideally could be integrated into the current efforts for
providing precise semantics for UML composite
structures~\cite{pscs-beta-1}.  Furthermore, we consider the
determinism of state machines to foster code
generation~\cite{derezinska-szczykulski:cst:2012}.

The remainder of this paper is structured as follows: In
Sect.~\ref{sec:hiuml} we provide some background on our goal of
heterogeneous institution-based UML semantics and introduce a small
example illustrating behavioural and protocol UML state machines.  In
Sect.~\ref{sec:UML-SM} we define institutions for these variants of
state machines. We study a notion of determinism for state machines,
their interleaving, and their refinement based on the institutions in
Sect.~\ref{sec:refinement}.  Finally, in Sect.~\ref{sec:conclusions}
we conclude with an outlook to future work.

\section{Heterogeneous Institution-based UML Semantics}
\label{sec:hiuml}

The work in this paper is part of a larger
effort~\cite{knapp-mossakowski-roggenbach:corr:2014} of giving an
institution-based heterogeneous semantics to several UML diagrams as
shown in Fig.~\ref{fig:languages}.
\begin{figure}[!Hb]
\centering
\begin{tikzpicture}[transform shape,scale=.86,auto,every loop/.style={-latex'}]
\tikzstyle{ingredient}=[shape=rectangle,draw,align=center,inner sep=0cm,outer sep=0cm,text width=3.2cm,minimum height=.8cm,minimum width=2.8cm,font={\sffamily\fontsize{8pt}{8pt}\selectfont}]
\tikzstyle{shorter}=[align=center,text width=2.4cm,minimum width=2.4cm,font={\sffamily\fontsize{8pt}{8pt}\selectfont}]
\tikzstyle{annotation}=[font={\sffamily\bfseries\fontsize{8pt}{8pt}\selectfont}]
\tikzstyle{consistency}=[draw,latex'-latex',color=black]
\tikzstyle{translation}=[draw,-latex',color=black,line width=2pt]
\draw (0, 0) node[ingredient,shorter] (mscs) {Interactions};
\draw (.4, 0) node[ingredient,draw=none] (propertiesV) {};
\draw (.4, 0) node[ingredient,draw=none] (mscsV) {};
\draw (0, -1) node[ingredient,shorter] (psms) {Protocol State\\ Machines};
\draw (.4, -1) node[ingredient,draw=none] (psmsV) {};
\draw (0, -2) node[ingredient,shorter,draw=none] (ocl) {Object Constraint\\ Language (OCL)};
\draw (.4, -2) node[ingredient,draw=none] (oclV) {};
\path[draw] (ocl.north west) |-
            ($ (ocl.south east) + (.8, 0) $) -|
            ($ (mscs.north east) + (.8, 0) $) -|
            ($ (ocl.north east) + (.21, 0) $) --
            cycle;
\draw (4.4, 0) node[ingredient] (sms) {State Machines};
\draw (4.4, 0) node[ingredient,draw=none] (typesV) {};
\draw (4.4, -1) node[ingredient] (sd) {Component Diagram};
\draw (4.4, -2) node[ingredient] (cd) {Class Diagram};
\draw (8.4, 0) node[ingredient,dashed] (smsi) {State Machine\\ Instances};
\draw (8.4, -1) node[ingredient] (csd) {Composite Structure\\ Diagram};
\draw (8.4, 0) node[ingredient,draw=none] (instancesV) {};
\draw (8.4, -2) node[ingredient] (od) {Object Diagram};
\draw (4.4, -1) node[shape=rectangle,draw,minimum height=3cm,minimum width=11.4cm] (uml) {};
\draw (uml.north east) node[annotation,anchor=north west] {Modelling in UML};
\draw (.4, -3.2) node[ingredient] (acsl) {ACSL};
\draw (6.4, -3.2) node[ingredient,minimum width=7.2cm] (c) {C};
\draw ($ (c.north east) + (0.1, 0.0) $) node[annotation,anchor=north west] {Implementation};
\draw ($ (propertiesV) + (0, 0.8) $) node[annotation,anchor=base] {Properties};
\draw ($ (typesV) + (0, 0.8) $) node[annotation,anchor=base] {Types};
\draw ($ (instancesV) + (0, 0.8) $) node[annotation,anchor=base] {Instances};
\end{tikzpicture}
\caption{Languages and diagrams to be considered}
\label{fig:languages}
\end{figure}
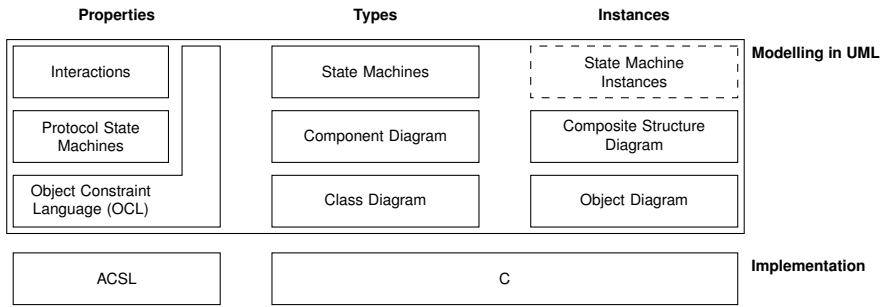
The vision is to provide semantic foundations for model-based
specification and design using a heterogeneous framework based on
Goguen's and Burstall's theory of institutions
\cite{GoguenBurstall92}.  We handle the complexity of giving a
coherent semantics to UML by providing several institutions
formalising different diagrams of UML, and several institution
translations (formalised as so-called institution morphisms and
comorphisms) describing their interaction and information flow. The
central advantage of this approach over previous approaches to formal
semantics for UML (e.g., \cite{lano:uml2:2009}) is that each UML
diagram type can stay ``as-is'', without the need of a coding using
graph grammars (as in \cite {DBLP:conf/uml/EngelsHK03}) or some logic
(as in \cite{lano:uml2:2009}). This also keeps full flexibility in the
choice of verification mechanisms.  The formalisation of UML diagrams
as institutions has the additional benefit that a notion of refinement
comes for free, see \cite{caslref,CASL-refinement-journal}.
Furthermore, the framework is flexible enough to support various
development paradigms as well as different resolutions of UML's
semantic variation points.  This is the crucial advantage of the
proposed approach to the semantics of UML, compared to existing
approaches in the literature which map UML to a specific global
semantic domain in a fixed way.


\subsection{Institutions}

Institutions are an abstract formalisation of the notion of logical
system.  Informally, institutions provide four different logical
notions: signatures, sentences, models and satisfaction.  Signatures
provide the vocabulary that may appear in sentences and that is
interpreted in models.  The satisfaction relation determines whether a
given sentence is satisfied in a given model. The exact nature of
signatures, sentences and models is left unspecified, which leads to a
great flexibility. This is crucial for the possibility to model UML
diagrams (which in the first place are not ``logics'') as
institutions.


More formally~\cite{GoguenBurstall92}, an institution $\institution{I}
= (\instSig{I}, \instSen{I}, \instMod{I}, {\instmodels{I}})$ consists
of (i)~a category of \emph{signatures} $\instSig{I}$; (ii)~a
\emph{sentence functor} $\instSen{I} : \instSig{I} \to
\category{Set}$, where $\category{Set}$ is the category of sets;
(iii)~a contra-variant \emph{model functor} $\instMod{I} :
(\instSig{I})\op \to \category{Class}$, where $\category{Class}$ is
the category of classes; and (iv)~a family of \emph{satisfaction
  relations} ${\instmodels[\Sigma]{I}} \subseteq |\instMod{I}(\Sigma)|
\times \instSen{I}(\Sigma)$ indexed over $\Sigma \in |\instSig{I}|$,
such that the following \emph{satisfaction condition} holds for every
signature morphism $\sigma : \Sigma \to \Sigma'$ in $\instSig{I}$,
every sentence $\varphi \in \instSen{I}(\Sigma)$ and for every
$\Sigma'$-model $M' \in |\instMod{I}(\Sigma')|$:
\begin{equation*}
  \instMod{I}(\sigma)(M') \instmodels[\Sigma]{I} \varphi
\ \iff\ 
  M' \instmodels[\Sigma']{I} \instSen{I}(\sigma)(\varphi)
\ \text{.}
\end{equation*}
$\instMod{I}(\sigma)$ is called the \emph{reduct} functor (also
written ${-}\reductop\sigma$), $\instSen{I}(\sigma)$ the
\emph{translation} function (also written $\sigma({-})$).

A \emph{theory} $T$ in an institution consists of a signature
$\Sigma$, written $\mathit{sig}(T)$, and a set of $\Sigma$-sentences;
its model class is the set of all $\Sigma$-models satisfying the
sentences.

An institution $\institution{I}$ has the
\emph{weak amalgamation property} for a pushout
\begin{equation*}
\begin{tikzpicture}[inner sep=0pt, outer sep=2pt]
  \matrix (m) [matrix of math nodes, ampersand replacement=\&, row sep=2.5em, column sep=2.5em, text height=1.7ex, text depth=0.25ex]{
    \Sigma   \& \Sigma_1\\
    \Sigma_2 \& \Sigma_R\\
  };
  \path[->,font=\scriptsize]
  (m-1-1) edge (m-2-1)
  (m-1-1) edge (m-1-2)
  (m-1-2) edge (m-2-2)
  (m-2-1) edge (m-2-2)
  ;
\end{tikzpicture}
\end{equation*}
if any pair $(M_1, M_2) \in \instMod{I}(\Sigma_1) \times
\instMod{I}(\Sigma_2)$ that is \emph{compatible} in the sense that
$M_1$ and $M_2$ reduce to the same $\Sigma$-model can be
\emph{amalgamated} to a $\Sigma_R$-model $M_R$ (i.e., there exists a
$M_R \in \instMod{I}(\Sigma_R)$ that reduces to $M_1$ and $M_2$,
respectively).

Weak amalgamation allows the computation of normal forms for
specifications \cite{Borzyszkowski:2002:LSS:639756.639759}, and
implies good behaviour w.r.t.  conservative extensions, as well as
soudness of proof systems for structured specifications
\cite{MossakowskiEtAl06}.

\subsection{ATM Example}

In order to illustrate our approach to a heterogeneous
institutions-based UML semantics in general and the institutions for
UML state machines in particular, we use as a small example the design
of a traditional automatic teller machine (ATM) connected to a
bank. For simplicity, we only describe the handling of entering a card
and a PIN with the ATM. After entering the card, one has three trials
for entering the correct PIN (which is checked by the bank). After
three unsuccessful trials the card is kept.

\begin{figure}[!Ht]
\centering
\subfigure[Interaction\label{fig:interaction}]{%
  \includegraphics[trim=6 6 6 6,clip,scale=0.65]{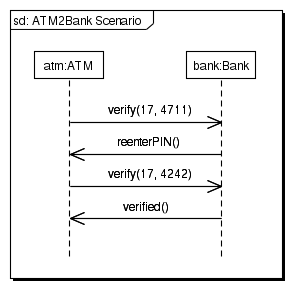}
}
\hspace*{0.5cm}
\subfigure[Composite structure\label{fig:system}]{%
  \includegraphics[trim=6 6 6 6,clip,scale=0.65]{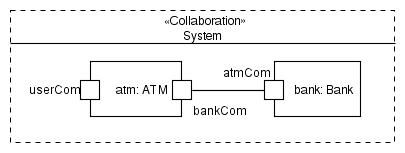}
}
\\
\subfigure[Protocol state machine\label{fig:psm}]{%
  \includegraphics[trim=6 6 6 6,clip,scale=.65]{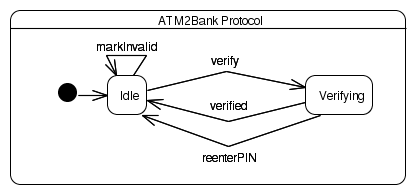}
}
\hspace*{0.4cm}
\subfigure[Interface\label{fig:class}]{%
  \includegraphics[trim=6 6 6 6,clip,scale=0.65]{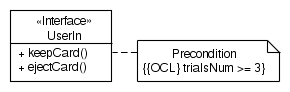}
}
\\
\subfigure[State machine\label{fig:state-machine}]{%
  \includegraphics[trim=6 6 6 6,clip,scale=0.65]{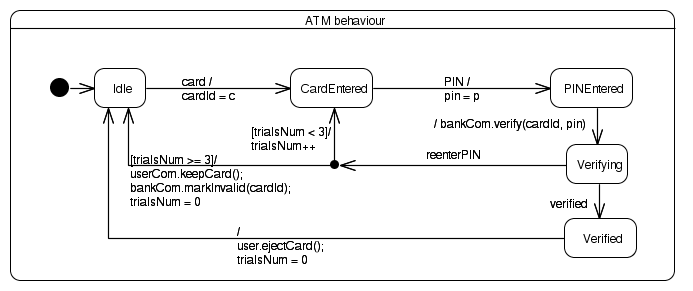}
}
\vspace*{-1.5ex}
\caption{ATM example}\label{fig:atm-example}
\end{figure}

Figure~\ref{fig:interaction} shows a possible \emph{interaction}
between an \uml{atm} and a \uml{bank} object, which consists out of
four messages: the \uml{atm} requests the \uml{bank} to \uml{verify}
if a card and PIN number combination is valid, in the first case the
\uml{bank} requests to reenter the PIN, in the second case the
verification is successful.  This interaction presumes that the system
has an \uml{atm} and a \uml{bank} as objects. This can, e.g., be
ensured by a \emph{composite structure diagram}, see
Fig.~\ref{fig:system}, which --- among other things --- specifies the
objects in the initial system state.  Furthermore, it specifies that
the communication between \uml{atm} and \uml{bank} goes through the
two ports \uml{bankCom} and \uml{atmCom} linked by a connector.  The
communication protocol on this connector is captured with a
\emph{protocol state machine}, see Fig.~\ref{fig:psm}.  The protocol
state machine fixes in which order the messages \uml{verify},
\uml{verified}, \uml{reenterPIN}, and \uml{markInvalid} between
\uml{atm} and \uml{bank} may occur.  Figure~\ref{fig:class} provides
structural information in form of an interface specifying what is
provided at the \uml{userCom} port of the \uml{atm} instance. An
interface is a set of operations that other model elements have to
implement. In our case, the interface is described in a \emph{class
  diagram}. Here, the operation \uml{keepCard} is enriched with the
OCL constraint \uml{trialsNum >= 3}, which refines its semantics:
\uml{keepCard} can only be invoked if the OCL constraints holds.

Finally, the dynamic behaviour of the \uml{atm} object is specified by
the \emph{behavioural state machine} shown in
Fig.~\ref{fig:state-machine}. The machine consists of five states
including \uml{Idle}, \uml{CardEntered}, etc.  Beginning in the
initial \uml{Idle} state, the user can \emph{trigger} a state change
by entering the \uml{card}. This has the \emph{effect} that the
parameter \uml{c} from the \uml{card} event is assigned to the
\uml{cardId} in the \uml{atm} object (parameter names are not shown on
triggers). Entering a \uml{PIN} triggers another transition to
\uml{PINEntered}.  Then the ATM requests verification from the bank
using its \uml{bankCom} port.  The transition to \uml{Verifying} uses
a \emph{completion event}: No explicit trigger is declared and the
machine autonomously creates such an event whenever a state is
completed, i.e., all internal activities of the state are finished (in
our example there are no such activities).  If the interaction with
the bank results in \uml{reenterPIN}, and the \emph{guard}
\uml{trialsNum < 3} is true, the user can again enter a \uml{PIN}.

\paragraph{Questions on the model.} 
Given the above diagrams specifying \emph{one} system, the question
arises if they actually ``fit'' together. Especially, one might ask if
the diagrams are consistent, and if the different levels of
abstraction refine each other. In our ATM example we have:

\begin{example}[Consistency]
The interface in Fig.\ \ref{fig:class} requires the operation
\uml{keepCard} only to be invoked when the precondition \uml{trialsNum
  >= 3} holds. This property holds for the state machine in
Fig.~\ref{fig:state-machine} thanks to the guard \uml{trialsNum < 3}.
\end{example}

\begin{example}[Refinement]
As the only trace of the interaction in Fig.\ \ref{fig:interaction}
is a possible run of the state machine in
Fig.\ \ref{fig:state-machine}, the interaction refines to the state
machine.
\end{example}

\begin{example}[Refinement]\label{ex:ref}
 Similarly, we can consider if the protocol state machine in
 Fig.\ \ref{fig:psm} refines to the product of the state machine of
 the \uml{atm}, shown in Fig.\ \ref{fig:state-machine}, and of the
 \uml{bank}; this essentially means to check for a trace inclusion
 w.r.t.\ messages observable on the interfaces, as the protocol state
 machine has no post conditions.
\end{example}

In order to study, e.g., such a refinement between a protocol state
machine and its implementation by state machines, in the following we
develop institutions for state machines including a notion of product.

\section{Institutions for Simple UML State Machines}
\label{sec:UML-SM}

We now detail a possible formalisation of a simplified version of UML
state machines as institutions. In particular, we omit hierarchical
states.  We start with institutions for the \emph{guards} and the
\emph{actions} of a state machine.  These fix the conditions which can
be used in guards of transitions, the actions for the effects of
transitions, and also the messages that can be sent from a state
machine.  The source of this information typically is a class or a
component diagram: The conditions and actions involve the properties
available in the classes or components, the messages are derived from
the available signals and operations.  The sentences of the action
institution form a simple dynamic logic (inspired by OCL) which can
express that if a guard holds as pre-condition, when executing an
action, a certain set of messages has been sent out, and another guard
holds as post-condition.  We then build a family of institutions for
\emph{state machines} over the institutions for guards and actions.  A
state machine adds the events and states that are used.  The events
comprise the signals and operations that can be accepted by the
machine; some of these will, in general, coincide with the messages
from the environment.  Additionally, the machine may react to
completion events, i.e., internal events that are generated when a
state of the machine has been entered and which trigger those
transitions that do not show an explicit event as their trigger in the
diagrammatic representation (we use the states as the names of these
events). %
\todoall{Rolf Hennicker bemerkte, dass die completion events ja gar
  keine seien, solange man nicht richtige Aktivitäten (oder gar
  activity diagrams) in den Zuständen hätten. Später relativierte er
  das aber: es könnte ja Transitionen mit normalem Trigger und
  triggerlose Transitionen geben, für letztere braucht man dann auch
  schon im einfachen Fall completion events. Wenn jedoch beide
  miteinander konkurrieren, ist die Priorität nicht ganz klar. HugoRT
  gibt der triggerlose Transition, also dem completetion event die
  Priorität. Rolf meinte aber, dass man das auch anders machen
  könne.}%
\todoalexander{Ich denke aber, daß obiger Fall, daß aus einem Zustand
  sowohl eine "echte" Transition, also eine mit einem "echten"
  Ereignis, als auch eine Completion transition, also mit Completion
  event, 'rausgehen, tatsächlich auftreten dürfte, und zwar, wenn die
  Completion transition einen Guard hat.  (Das könnte eventuell auch
  beim "Flachklopfen" einer hierarchischen Zustandsmaschine
  automatisch auftreten - muß ich mir aber noch genauer überlegen).
  Laut UML wird das Completion event immer bevorzugt (immer), die
  Completion transition könnte aber durchaus nicht enabled sein, da
  der Guard momentan nicht wahr ist.}%
The initial state as well as the transitions of the machine are
represented as sentences in the institution.\footnote{For simplicity,
final states are left implicit here. For hierarchical states, they
need to be made explicit.} %
\todoall{What about final states? Are they given implicitly as the set
  of states without an outgoing transition? Do they play a role? In a
  sense, they should\ldots Maybe we can briefly discuss runs, both of
  the LTS in the sentence, and of the semantic LTS. AK: Yes, final
  states are currently missing.  For flat state machines they are not
  too important, but they would definitely make a difference for runs
  (getting stuck vs.\ successful termination).}%
In a next step, we combine the family of state machine institutions
parameterised over actions into a single institution.%
\todoall{We could claim that there is a co-morphism from the
  protocol state machine to the environment institution and/or the OCL
  institution --- but this seems to be quite bold.}



\subsection{Institution of Guards}

We assume that there is an institution of guards. Typically, guards
are formulas in some language like OCL.  More formally, an
\emph{institution of guards} is an institution where signatures are
sets, and signature morphisms are functions.  (We will call the
elements of these sets variables, but one can think of attributes,
operations and signals being collected here as well.)  Models of a
signature $V$ are valuations $\omega : V \to \Val$ into a fixed set of
values $\Val$\footnote{In UML, variables and values would be typed,
  and variable valuations have to respect the typing. For simplicity,
  we disregard this here. Moreover, for operations, valuations would
  assign values in some function space.}. Model reduct is just
composition, that is, given a signature morphism $v : V \to V'$ and a
model $\omega' : V' \to \Val$, its $v$-reduct is $\omega' \circ
v$. The nature of sentences $G(V)$ and their translation $G(v) : G(V)
\to G(V')$ is left unspecified, as well as the satisfaction relation
--- we only require the satisfaction condition, which amounts to
\begin{equation*}
  \omega' \models G(v)(g)
\quad\text{iff}\quad
  \omega' \circ v \models g
\ \text{.}
\end{equation*}

\begin{example}
Consider the UML component \uml{ATM}.  An guard signature for
\uml{ATM} would contain the variable \uml{trialsNum}, leading to
sentences such as \uml{true}, \uml{trialsNum < $n$}, and
\uml{trialsNum == $n$} for $n \in \NZ$.
\end{example}

\subsection{Institution of Actions}

An object of the category of action \emph{signatures} %
$\instSig{\ENV}$ is a triple of sets
\begin{equation*}
  \Eta = (A_{\Eta}, M_{\Eta},V_{\Eta})
\end{equation*}
of actions, messages and variables; and a morphism $\Eta \to \Eta'$ of
$\instSig{\ENV}$ is a triple of functions $\eta : (\eta_A : A_{\Eta}
\to A_{\Eta'}, \eta_M : M_{\Eta} \to M_{\Eta'}, \eta_V : V_{\Eta} \to
V_{\Eta'})$. 
The class of action \emph{structures} $\instMod{\ENV}(\Eta)$ for an
action signature $\Eta$ consists of transition relations
\begin{equation*}
  \Omega \subseteq  |\Omega| \times (A_{\Eta} \times \powerset{M_{\Eta}}) \times |\Omega|
\ \text{,} 
\end{equation*}
where $|\Omega|=(V_\Eta\rightarrow \Val)$ represents the possible
configurations of data states,
and
\begin{equation*}
  (\omega,a,\overline{m},\omega')\in\Omega
\quad\text{(also written $\shorttrans{\omega}{a, \overline{m}}{\Omega}{\omega'}$)}
\end{equation*}
expresses that action $a$ leads from state $\omega \in
(V_\Eta\rightarrow \Val)$ to state $\omega' \in
(V_\Eta\rightarrow \Val)$ producing the set of messages
$\overline{m} \subseteq M_{\Eta}$.

The \emph{reduct} $\Omega'\reductop\eta$ of an $\Eta'$-action
structure $\Omega'$ along the morphism $\eta : \Eta \to \Eta'$ is
given by all transitions
\begin{equation*}
  \xtrans{\omega_1\reductop_{\eta_V}}{a,\eta_M^{-1}(\overline{m})}{\Omega'\reductop_\eta}{\omega_2\reductop_{\eta_V}}
\quad\text{for which}\quad
  \xtrans{\omega_1}{\eta_A(a),\overline{m}}{\Omega'}{\omega_2} 
\ \text{.}
\end{equation*}

An action $a$ is called \emph{deterministic} if
$\xtrans{\omega_1}{a,\overline{m}}{\Omega}{\omega_2}$ and
$\xtrans{\omega_1}{a,\overline{m'}}{\Omega}{\omega'_2}$ imply
$\overline{m}=\overline{m'}$ and $\omega_2=\omega'_2$.  An action
relation $\Omega$ is called deterministic if all its actions are
deterministic, that is, it is a partial function of type $|\Omega|
\times A_{\Eta} \rightharpoonup \powerset(M_{\Eta}) \times |\Omega|$.

Note that reducts can introduce non-determinism. Given an action
signature $\Eta$ with $V_\Eta = \{ x, y \}$, suppose that a
deterministic action $a$ leads to a change of state expressed by the
assignment $x:=x+y$.  Now take the reduct to the same signature but
with $V_\Eta=\{x\}$, i.e., the variable $y$ has been removed. Then $a$
performs a non-deterministic assignment $x:=x+y$ where the value for
$y$ is non-deterministically guessed.

The set of action \emph{sentences} $\instSen{\ENV}(\Eta)$ for an
action signature $\Eta$ comprises the expressions
\begin{equation*}
  g_{\mathrm{pre}} \rightarrow [a]\overline{m} \rhd g_{\mathrm{post}}
\end{equation*}
with $g_{\mathrm{pre}},\allowbreak g_{\mathrm{post}} \in G(V_{\Eta})$,
$a \in A_{\Eta}$, and $\overline{m} \subseteq M_{\Eta}$, intuitively
meaning (like an OCL constraint) that if the pre-condition
$g_{\mathrm{pre}}$ currently holds, then, after executing $a$, the
messages $\overline{m}$ are produced and the post-condition
$g_{\mathrm{post}}$ holds.  The \emph{translation}
$\eta(g_{\mathrm{pre}} \rightarrow [a]\overline{m} \rhd
g_{\mathrm{post}})$ of a sentence $g_{\mathrm{pre}} \rightarrow
[a]\overline{m} \rhd g_{\mathrm{post}}$ along the signature morphism
$\eta : \Eta \to \Eta'$ is given by $G(\eta_V)(g_{\mathrm{pre}})
\rightarrow [\eta_A(a)]\eta_M(\overline{m}) \rhd
G(\eta_V)(g_{\mathrm{post}})$.  Finally, the satisfaction relation
$\Omega \instmodels[\Eta]{\ENV} g_{\mathrm{pre}} \rightarrow
[a]\overline{m} \rhd g_{\mathrm{post}}$ holds if, and only if, for all
$\omega \in (V_\Eta\rightarrow \Val)$, if $\omega \models
g_{\mathrm{pre}}$ and $\shorttrans{\omega}{ a, \overline{m}'}{\Omega}{
  \omega'}$, then $\omega' \models g_{\mathrm{post}}$ and $\overline{m}
\subseteq \overline{m}'$.  Then the \emph{satisfaction condition} follows.


\begin{example}
Consider the UML component \uml{ATM} with its properties \uml{cardId},
\uml{pin}, and \uml{trialsNum}, its ports \uml{userCom} and
\uml{bankCom}, and its outgoing operations \uml{ejectCard()} and
\uml{keepCard()} to \uml{userCom}, and \uml{verify()} and
\uml{markInvalid()} to \uml{bankCom}.  An action signature for
\uml{ATM} is derived by forming actions and messages over
this information, such that it will contain the
actions \uml{user.ejectCard(); trialsNum = 0} and \uml{trialsNum++},
as well as the messages \uml{user.ejectCard()} and
\uml{bank.markInvalid(cardId)}.  Action sentences over such an
action signature could be
\begin{gather*}
  \uml{true} \rightarrow [\uml{user.ejectCard(); trialsNum = 0}] \{ \uml{user.ejectCard()} \} \rhd \uml{trialsNum == 0}
\quad\text{or}\\
  \uml{trialsNum == $n$} \rightarrow [\uml{trialsNum++}]\emptyset \rhd{} \uml{trialsNum == $n$+1}
\ \text{.}
\end{gather*}
\end{example}

\subsection{Behavioural State Machine Institution}

The institution of state machines is now built over the action
institution.  Let $\Eta$ be an action signature and $\Omega$ an
action structure over $\Eta$.  An object of the category of state
machine \emph{signatures} $\instSig{\SM[(\Eta, \Omega)]}$ over $\Eta$
and $\Omega$ is given by a triple
\begin{equation*}
  \Sigma = (E_{\Sigma}, F_{\Sigma}, S_{\Sigma})
\end{equation*}
of (external) events $E_{\Sigma}$, completion events $F_{\Sigma}$, and
states $S_{\Sigma}$ with $E_{\Sigma} \cap F_{\Sigma} = \emptyset$ and
$E_{\Sigma} \cap S_{\Sigma} = \emptyset$; and a morphism $\sigma :
\Sigma \to \Sigma'$ of $\instSig{\SM[(\Eta, \Omega)]}$ is a triple of
injective functions $\sigma = (\sigma_E : E_{\Sigma} \to E_{\Sigma'},
\sigma_F : F_{\Sigma} \to F_{\Sigma'}, \sigma_S : S_{\Sigma} \to
S_{\Sigma'})$, such that $E_\Sigma\cap M_\Eta=E_\Sigma'\cap M_\Eta$ (preservation
of internal messages).  The class of state machine \emph{structures}
$\instMod{\SM[(\Eta, \Omega)]}(E_{\Sigma}, F_{\Sigma}, S_{\Sigma})$
for a state machine signature $(E_{\Sigma}, F_{\Sigma}, S_{\Sigma})$
over $\Eta$ and $\Omega$ consists of the pairs
\begin{equation*}
  \Theta = (I_{\Theta}, \Delta_{\Theta})
\end{equation*}
where $I_{\Theta} \in \powerset (V_\Eta\rightarrow \Val) \times
S_{\Sigma}$ represents the initial configurations, fixing the initial
control state; and $\Delta_{\Theta} \subseteq C \times
\powerset(M_{\Eta}) \times C$ with $C = (V_\Eta\rightarrow \Val)
\times \powerset(E_\Sigma \cup F_{\Sigma}) \times S_{\Sigma}$
represents a transition relation from a configuration, consisting of
an action state, an event pool, and a control state, to a
configuration, emitting a set of messages.  The event pool may contain
both types of events from the signature: external events from signals
and operations, and completion events (which are typically represented
by states).

\begin{example}
Consider the state machine of Fig.~\ref{fig:state-machine} defining
the behaviour of \uml{ATM}.  It works over the action signature
sketched in the previous example, and its signature is
$(E_{\uml{ATM}}, F_{\uml{ATM}}, S_{\uml{ATM}})$ with
\begin{gather*}
  E_{\uml{ATM}} = \{ \uml{card}, \uml{PIN}, \uml{reenterPIN}, \uml{verified} \}
\ \text{,}\\
  F_{\uml{ATM}} = \{ \uml{PINEntered}, \uml{Verified} \}
\ \text{,}\\
  S_{\uml{ATM}} = \{ \uml{Idle}, \uml{CardEntered}, \uml{PINEntered}, \uml{Verifying}, \uml{Verified} \}
\ \text{.}
\end{gather*}
In particular, the completion events consist of those states from
which a completion transition originates.
\end{example}

The \emph{reduct} $\Theta'\reductop\sigma$ of a state machine structure $\Theta'$ along the morphism $\sigma :
\Sigma \to \Sigma'$ is given by the structure
\begin{equation*}
  (\{ (\omega, s) \mid (\omega, \sigma_S(s)) \in I' \}, \Delta)
\end{equation*}
where $\sigma_P(p) = \sigma_E(p)$ if $p \in E_{\Sigma}$ and $\sigma_P(p) =
\sigma_F(p)$ if $p \in F_{\Sigma}$, and $\Delta$ is given by
\begin{equation*}
  \{\shorttrans{(\omega_1,\sigma_P^{-1}(\overline{p_1}),s_1)}{\overline{m}}{}{(\omega_2,\sigma_P^{-1}(\overline{p_2}),s_2)}
\mid
  \shorttrans{(\omega_1,\overline{p_1},\sigma_S(s_1))}{\overline{m}}{\Delta'}{(\omega_2,\overline{p_2},\sigma_S(s_2))} \}
\ \text{.}
\end{equation*}
Here, $\sigma_P^{-1}$ deletes those events from the event pool that
are not present in the pre-image.

The set of state machine \emph{sentences} $\instSen{\SM[(\Eta,
  \Omega)]}(\Sigma)$ for a state machine signature $\Sigma$ over
$\Eta$ and $\Omega$ consists of the pairs
\begin{equation*}
 \varphi = (s_0 \in S_{\Sigma}, T \subseteq S_{\Sigma} \times (E_{\Sigma} \cup F_{\Sigma}) \times (G(V_{\Eta}) \times A_{\Eta} \times \powerset(F_{\Sigma})) \times S_{\Sigma})
\end{equation*}
where $s_0$ means an initial state and the prioritised set $T$
represents the transitions from a state $s$ with a triggering event
$p$ (either a declared event or a completion event), a guard $g$, an
action $a$, and a set of completion events $\overline{f}$ to another
state $s'$.  We also write $\xtrans{s}{p[g]/a, \overline{f}}{T}{s'}$
for such a transition.  The translation $\sigma(s_0, T)$ of a sentence
$(s_0, T)$ along the signature morphism $\sigma : \Sigma \to \Sigma'$
is given by $(\sigma_S(s_0),\allowbreak \{ \shorttrans{\sigma_S(s_1)}
{\sigma_P(p)[g]/a, \powerset\sigma_F(\overline{f})}{}{\sigma_S(s_2)}
\mid \shorttrans{s_1}{p[g]/a, \overline{f}}{T}{s_2} \})$.  Finally,
the \emph{satisfaction relation} $\Theta
\instmodels[\Sigma]{\SM[(\Eta, \Omega)]} (s_0, T)$ holds if, and only
if $\pi_2(I_{\Theta}) = s_0$ and $\Delta_{\Theta}$ is the least
transition relation satisfying\footnote {Usually, the two cases do not
  overlap, so the two cases are complete characterisations (iff).}
\begin{gather*}
  \shorttrans{(\omega, p :: \overline{p}, s)}{\overline{m} \setminus E_{\Sigma}}{\Delta_{\Theta}}{(\omega', \overline{p} \lhd ((\overline{m} \cap E_{\Sigma}) \cup \overline{f}), s')}
\quad\text{if}\\\qquad
  \exists \shorttrans{s}{p[g]/a, \overline{f}}{T}{s'} \,.\, \omega \models g \land \shorttrans{\omega}{a,\overline{m}}{\Omega}{\omega'}
\\
  \shorttrans{(\omega, p :: \overline{p}, s)}{ \emptyset}{\Delta_{\Theta}}{ (\omega, \overline{p}, s)}
\quad\text{if}\\\qquad
  \forall \shorttrans{s}{p'[g]/a, \overline{f}}{T}{s'} \,.\, p \neq p' \lor \omega \not\models g
\end{gather*}
where $p :: \overline{p}$ expresses that the first element $p$ from
the pool $\overline{p}$ is extracted, and $\overline{p} \lhd
\overline{p}'$ adds the events in $\overline{p}'$ to the pool
$\overline{p}$ with respect to selection scheme (where completion
events are prioritised).  The messages on a transition in the
structure $\Theta$ are only those that are not accepted by the machine
itself, i.e., not in $E_{\Sigma}$.  The accepted events in
$E_{\Sigma}$ as well as the completion events are added to the event
pool of the target configuration.  When no transition is triggered by
the current event, the event is discarded (this will happen, in
particular, to all superfluously generated completion
events). Checking the satisfaction condition
\begin{equation*}
  \Theta'\reductop\sigma \instmodels[\Sigma]{\SM[(\Eta, \Omega)]} (s_0, T)
\ \iff\ 
  \Theta \instmodels[\Sigma']{\SM[(\Eta, \Omega)]} \sigma(s_0, T)
\end{equation*}
for a state machine signature morphism $\sigma : \Sigma \to \Sigma'$
is straightforward.

\begin{example}
Continuing the previous example for the state machine of
Fig.~\ref{fig:state-machine} defining the behaviour of \uml{ATM}, this
state machine can be represented as the following sentence over this
signature:
\begin{gather*}
  (\uml{Idle},\{
   \longtrans{\uml{Idle}}{\uml{card}[\uml{true}]/\uml{cardId = c}, \emptyset}{T}{\uml{CardEntered}},
\\\phantom{(\uml{Idle}, \{}
   \longtrans{\uml{CardEntered}}{\uml{PIN}[\uml{true}]/\uml{pin = p}, \uml{PINEntered}}{T}{ \uml{PINEntered}},
\\\phantom{(\uml{Idle}, \{}
   \lllongtrans{\uml{PINEntered}}{\uml{PINEntered}[\uml{true}]/\uml{bank.verify(cardId, pin)}, \emptyset}{T}{ \uml{Verifying}},
\\\phantom{(\uml{Idle}, \{}
   \lllongtrans{\uml{Verifying}}{\uml{reenterPIN}[\uml{trialsNum < 3}]/\uml{trialsNum++}, \emptyset}{T}{ \uml{CardEntered}},
   \ldots \})
\ \text{.}
\end{gather*}
In particular, \uml{PINEntered} occurs both as a state and as a
completion event to which the third transition reacts.  The junction
pseudostate for making the decision whether \uml{trialsNum < 3} or
\uml{trialsNum >= 3} has been resolved by combining the transitions.
\end{example}

\subsection{Protocol State Machine Institution}

Protocol state machines differ from behavioural state machines by not
mandating a specific behaviour but just monitoring behaviour: They do
not show guards and effects, but a pre- and a postcondition for the
trigger of a transition.  Moreover, protocol state machines do not
just discard an event that currently does not fire a transition; it is
an error when such an event occurs.

For adapting the state machine institution to protocol state machines
we thus change the \emph{sentences} to
\begin{equation*}
  \varphi = (s_0,e \in S_{\Sigma}, T \subseteq S_{\Sigma} \times (G(V_{\Eta}) \times E_{\Sigma} \times G(V_{\Eta}) \times \powerset(M_{\Eta}) \times \powerset(F_{\Sigma})) \times S_{\Sigma})
\end{equation*}
where $s_0$ is the start state and $e$ a dedicated error state,
 the two occurrences of $G(V_{\Eta})$ represent the pre- and the
post-conditions, and $\powerset(M_{\Eta})$ represents the messages
that have to be sent out in executing the triggering event (protocol
state machines typically do not show completion events).  The
\emph{satisfaction relation} now requires that when an event $e$ is
chosen from the event pool the pre-condition of some transition holds
in the source configuration, its post-condition holds in the target
configuration, and that all messages have been sent out.  Instead of
the second clause of $\Delta_{\Theta}$, discarding an event, 
the error state is targeted when no transition is enabled.

\subsection{Flat State Machine Institution}

Fix an institution of guards.  We now flatten the institutions
$\SM[(\Eta, \Omega)]$ for each action signature $\Eta$ and each action
structure $\Omega$ over $\Eta$ into a single institution
$\SM$.\footnote {This is an instance of a general construction, namely
  the Grothendieck institution \cite{Diaconescu02}.}  The signatures
$\langle\Eta, \Sigma\rangle$ consist of an action signature $\Eta$ and
a state machine signature $\Sigma$, similarly for signature morphisms
as well as for structures $\langle\Omega, \Theta\rangle$.  As
$\langle\Eta, \Sigma\rangle$-sentences we now have both dynamic logic
formulas (over $\Eta$), as well as control transition relations (over
$\Eta$ and $\Sigma$).  Also satisfaction is inherited.  Only the
definition of reducts is new, because they need to reduce state
machine structures along more complex signature morphisms:
$\langle\Omega', \Theta'\rangle\reductop(\eta, \sigma) =
\langle\Omega'\reductop\eta,
\Theta'\reductop\sigma\reductop\eta\rangle$ where
$\Theta''\reductop\eta = (I_{\Theta''}, \{ c''_1,
\eta_M^{-1}(\overline{m}''), c''_2) \mid (c''_1, \overline{m}'',
c''_2) \in \Delta_{\Theta''} \})$.

\section{Determinism, Interleaving, and Refinement}
\label{sec:refinement}

\subsection{Determinstic State Machines}

The transition and action relations are not required to be
functions. Thus a transition may have multiple choices for the same
configuration of states, variables and events. But when moving towards
the implementation, deterministic behaviour is desirable.
\begin{enumerate} 
  \item[i-a)] A prioritised transition set $T$ is called
\emph{syntactically deterministic} if it is a partial function of type
$S_{\Sigma} \times (E_{\Sigma} \cup F_{\Sigma}) \rightharpoonup
G_{\Eta} \times A_{\Eta} \times \powerset(F_{\Sigma}) \times S_{\Sigma}$.

  \item[i-b)] A prioritised transition set $T$ is called
\emph{semantically deterministic} if for any two distinct transitions
$\shorttrans{s}{p[g_1]/a_1, \overline{f}_1}{T}{s_1}$ and
$\shorttrans{s}{p[g_2]/a_2, \overline{f}_2}{T}{s_2}$ sharing the same
pre-state $s$ and trigger event $p$, their guards must be disjoint,
that is, there is no $\omega : V_{\Eta} \to \Val$ satisfying both
$g_1$ and $g_2$.

  \item[ii)] A transition relation $\Delta_\Theta$ is called
\emph{deterministic} if and only if it is a partial function of type
$C \rightharpoonup \powerset(M_{\Eta}) \times C$.\footnote{Note that
  this function is total if $\Delta_\Theta$ satisfies some sentence.
  This originates from the discarding of events that can not be
  processed in the current state and configuration, which is again a
  transition.}

  \item[iii)] A state machine $(\Omega, \Theta)$ is called
\emph{deterministic} if and only if the corresponding transition
relation and action relation are deterministic.

\end{enumerate}

The transition relation $\Delta_\Theta$ is defined by $\Omega$ and
$T$. So it is justified to expect some inheritance of determinism
between those.

\begin{theorem}
If $T$ is syntactically or semantically deterministic and $\Omega$ is
deterministic, then $\Delta_\Theta$ is also deterministic.
\end{theorem}
\begin{proof}
Consider a configuration $(\omega, p::\overline{p}, s)$. If there is
any transition then the new state $s'$ and executed action $a$ are
determined by $T(s, p) = (g, a, \overline{f}, s')$ (if defined) in the
syntactic case. The sent message $\overline{m}$ and the new
configuration of the variables $\omega'$ result from $\Omega((\omega,
p, s), a)=(\overline{m}, \omega')$. In the semantic case, at most one
guard can be enabled, hence at most one transition in $T$ can fire.
\end{proof}


\subsection{Interleaving Product of State Machines}

Inside the flat state machine institution $\SM$ we can consider the
composition of state machines over different action signatures.  The
composition captures the interplay between different state machines
and their communication.  The different action signatures represent
the local views of the state machines.  Given two state machine
signatures $\langle\Eta_1, \Sigma_1\rangle$ and $\langle\Eta_2,
\Sigma_2\rangle$ of $\SM$ with $E_{\Sigma_1} \cap E_{\Sigma_2} =
\emptyset$, and $S_{\Sigma_1} \cap S_{\Sigma_2} = \emptyset$, we
combine these into a single signature $\langle\hat{\Eta},
\hat{\Sigma}\rangle$ of $\SM$ by taking the component-wise union for
the guard, actions, messages, and variables, the union of events and
states for the events, and the product of the state sets for the
states. Now consider two state machine structures $(\Omega_1,
\Theta_1)$ over $\langle\Eta_1, \Sigma_1\rangle$ and $(\Omega_2,
\Theta_2)$ over $\langle\Eta_2, \Sigma_2\rangle$, respectively.  Their
\emph{interleaving product} is given by
\begin{equation*}
  \langle\Omega_1, \Theta_1\rangle \pll \langle\Omega_2,
  \Theta_2\rangle = (\Omega_1 \pll \Omega_2, \Theta_1 \pll \Theta_2)
\qquad\text{where}
\end{equation*}
\begin{itemize}[label={--},topsep=2pt,itemsep=0pt,leftmargin=*]
  \item $\Omega_1 \pll \Omega_2$ is given by
$\shorttrans{\omega}{a,\overline{m}}{\Omega_1 \pll \Omega_2}{\omega'}$
if for some $i\in\{1,2\}$: $a \in A_{\Eta_i}$ and
$\shorttrans{\omega\reductop_{V_{\Eta_i}}}{a,\overline{m}}{\Omega_i}{\omega'\reductop_{V_{\Eta_i}}}$
and for $i \neq j \in \{ 1, 2 \} : \omega\reductop_{V_{\Eta_j}
  \setminus V_{\Eta_i}} = \omega'\reductop_{V_{\Eta_j} \setminus
  V_{\Eta_i}}$

  \item $\Theta_1 \pll \Theta_2 = (I_{\Theta_1} \pll I_{\Theta_2},
\Delta_{\Theta_1} \pll \Delta_{\Theta_2})$ with
\begin{equation*}
  I_{\Theta_1} \pll I_{\Theta_2}
=
  (\{\omega : V_{\Eta_1} \cup V_{\Eta_2} \to \Val \mid \forall j \in \{ 1, 2\} \,.\, \omega\reductop_{V_{\Eta_j}} \in \Gamma_j \}, (s_1, s_2))
\end{equation*}
for $I_{\Theta_i} = (\Gamma_i, s_i)$, and
\bgroup\mathindent0pt
\begin{gather*}
  \longtrans{(\omega, p :: (\overline{p}_1 \cup \overline{p}_2), (s_1, s_2))}
            {\overline{m} \setminus E_{\hat{\Sigma}}}
            {\Delta_{\Theta_1} \pll \Delta_{\Theta_2}}
            {(\omega', (\overline{p}_1 \cup \overline{p}_2) \lhd ((\overline{p}' \cup \overline{m}) \cap E_{\hat{\Sigma}}), (s_1', s_2'))}
\\\qquad
\begin{array}{@{}ll@{}}
\text{iff } & \exists i \in \{ 1, 2 \} \,.\, \shorttrans{(\omega\reductop_{V_{\Eta_i}}, p :: \overline{p}_i, s_i)}{  \overline{m}}{\Delta_{\Theta_i}}{(\omega'\reductop_{V_{\Eta_i}}, \overline{p}_i \lhd \overline{p}', s_i')} \land{}\\
            &\qquad \forall j \in \{ 1, 2 \} \setminus \{ i \} \,.\, (\omega\reductop_{V_{\Eta_j}}, \overline{p}_j, s_j)=(\omega'\reductop_{V_{\Eta_j}}, \overline{p}_j, s_j')
\ \text{.}
\end{array}
\end{gather*}
\egroup
\end{itemize}

There is also a syntactic version of the interleaving product: %
given sentences $(s_0^1, T_1)$ and $(s_0^2, T_2)$, their interleaving
$(s_0^1, T_1) \pll (s_0^2,T_2)$ is given by $((s_0^1,s_0^2), T)$ with
\begin{equation*}
  \shorttrans{(s_1, s_2)}{p[g]/a, \overline{f}}{T}{(s'_1, s'_2)}
\ \text{iff}\ 
  \exists i \in \{ 1, 2 \} \,.\, \shorttrans{s_i}{p[g]/a, \overline{f}}{T_i}{s'_i}
\land
  \forall j \in \{ 1, 2 \} \setminus \{ i \} \,.\, s_j = s'_j
\ \text{.}
\end{equation*}
The syntactic version is compatible with the semantic interleaving product:
\begin{theorem}\label{thm:compatible-interleaving}
If  $\langle\Omega_i, \Theta_i\rangle \models (s_0^i, T_i)$ ($i=1,2$),
then $\langle\Omega_1, \Theta_1\rangle \pll \langle\Omega_2,
  \Theta_2\rangle\models (s_0^1, T_1) \pll (s_0^2,T_2)$.
\end{theorem}

\begin{example}
Consider the composite structure diagram in Fig.~\ref{fig:system},
showing instances \uml{atm} and \uml{bank} of the \uml{ATM} and
\uml{Bank} components, respectively, that are connected through their
\uml{bankCom} and \uml{atmCom} ports.  In execution, \uml{atm} and
\uml{bank} will exchange messages, as prescribed by their state
machines, and this exchange is reflected by the interleaving product
which internalises those events that are part of the common signature.
On the other hand, messages to the outside, i.e., through the
\uml{userCom} port are still visible.
\end{example}

A system resulting from an interleaving product $\langle\Omega_1,
\Theta_1\rangle \pll \langle\Omega_2,\Theta_2\rangle$ represents a
state machine in our notation. Thus it can be again part of an
interleaving product
\begin{equation*}
  \langle(\langle\Omega_1, \Theta_1\rangle \pll \langle\Omega_2,\Theta_2\rangle)\pll \langle\Omega_3,\Theta_3\rangle \rangle
\ \text{.}
\end{equation*}

The interleaving product meets the intuitive algebraic properties. Due
to the disjoint event sets each event can only trigger at most one
machine. Messages are stripped off the events which can be processed
by either of the inner machines, and remaining messages are sent to
the third machine, which also extracts its corresponding events as
illustrated in Fig.~\ref{fig:assoc}.
\begin{figure}[!Ht]
\centering
\begin{tikzpicture}
\tikzstyle{surround} = [fill=blue!10,thick,draw=black,rounded corners=2mm]
\node[draw] (M1) at (1,0) {$M_1$};
\node[inner sep=0,minimum size=0] (k) at (2,0){}; 
\node[inner sep=0,minimum size=0] (k2) at (3,0){}; 
\node[inner sep=0,minimum size=0] (k3) at (3,1.6){}; 
\node[inner sep=0,minimum size=0] (k4) at (3,2.4){}; 
\node[draw] (M2) at (5,0) {$M_2$};
\node[draw] (M3) at (5,1.6) {$M_3$};
\node[draw] (M21) at (5,-1) {$(\Omega_1 \pll \Omega_2, \Theta_1 \pll \Theta_2)$};

\draw[-] (M1) -- (k) node[pos=0.5,below] (m) {$\overline{m}$};
\draw[->, bend right] (k) -- node[below right](bk){$\overline{m} \cap E_{\Sigma_1}$} ++(0,-1)  -| (M1);
\draw[-] (k) -- (k2) node[pos=0.5,above right]{};
\draw[->] (k) -- (M2) node[pos=0.5,below right]{$\overline{m} \cap E_{\Sigma_2}$};
\draw[-] (k2) -- (k3) node[pos=0.5,right]{$\overline{m} \backslash (E_{\Sigma_1} \cup E_{\Sigma_2})$};
\draw[->] (k3) -- (M3) node[pos=0.5,midway, above]{$\overline{m} \cap E_{\Sigma_3}$};
\draw[->,dashed] (k3) -- (k4) node[above]{$\overline{m} \backslash ((E_{\Sigma_1} \cup E_{\Sigma_2})  \cup E_{\Sigma_3})$};
\begin{pgfonlayer}{background} 
  \node[surround] (background) [fit = (M1) (M2) (bk) (M21)] {};
\end{pgfonlayer}
\end{tikzpicture}
\caption{Messages sent between three machines on transition in machine $M_1$}\label{fig:assoc}
\end{figure}
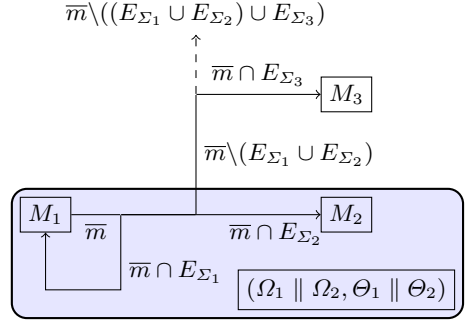  
Hence it is impossible that the inner machines consume an event that
can also be processed by the third machine. The same behaviour occurs,
if the first machine sends a message to a system of the two remaining
machines. Thus the distinction in `inner' and `outer' machines becomes
obsolete and the interleaving product is associative.  Since each
machine extracts the events present in its event set, it is not
required to consider the order of the machines, and hence the
interleaving product is commutative.  Finally, we can regard a state
machine with no events, no messages, and only one state. In an
interleaving product this nearly empty machine would have no effect on
the behaviour of the other machine, and thus behaves as a neutral
element.

\begin{theorem}\label{th:monoid}
The set of state machine structures (over all signatures) with
interleaving product $\pll$ forms a discrete symmetric monoidal
category, which is a ``commutative monoid up to isomorphism''.
\end{theorem}

It is desirable that the interleaving product of two deterministic
machines preserves this determinism.  The new action function is
determined by the two old ones in such a way, that the new
configuration is taken from the new configuration of the triggered
sub-machine, which is deterministic, and the missing variable
configuration remain untouched. Thus the new action relation is
deterministic.  The same goes for the transition relation. The sent
messages and configuration are determined by the (as argued above)
deterministic action relation and the new state and events result from
the triggered sub-machine. However, we need the following prerequisite:
 Two action relations $\Omega_1$, $\Omega_2$
are called \emph{compatible} if $\Omega_1\reductop_{\Eta_1\cap\Eta_2}=
\Omega_2\reductop_{\Eta_1\cap\Eta_2}$.

\begin{theorem}\label{th:InhDetIP}
Let $(\Omega_1, \Theta_1)$ and $(\Omega_2, \Theta_2)$ be deterministic
state machines with both action relations compatible. Then $(\Omega_1
\pll \Omega_2, \Theta_1 \pll \Theta_2)$ is also deterministic.
\end{theorem}

Using a slightly modified version of the interleaving product
construction where messages of shared actions leading to compatible
states are united, instead of generating two separate transitions, we
can prove:

\begin{theorem}\label{thm:amalg-action}
The action institution admits weak amalgamation for pushout squares
with injective message mappings.
\end{theorem}

\subsection{Institutional Refinement of State Machines}

We have defined an institution capturing both behavioural and protocol
state machines via different sentences.  With the machinery developed
so far, we can now apply the institution independent notion of
refinement to our institution of state machines.  The simplest such
notion is just model class inclusion, that is, a theory $T_1$ refines
to $T_2$ (written $T_1 \leadsto T_2$) if
$\instMod{\SM}(T_2)\subseteq\instMod{\SM}(T_1)$. (Note that state
machines are theories consisting typically of one sentence only.)

However, this is too simple to cover the phenomenon of \emph{state
  abstraction}, where several states (like \uml{Idle},
\uml{CardEntered}, \uml{PinEntered} and \uml{Verified} in
Fig.~\ref{fig:state-machine}) in a more concrete state machine can be
abstracted to one state (like \uml{Idle} in Fig.~\ref{fig:psm}) in a
more abstract state machine. This situation can be modelled using the
institution independent notion of \emph{translation} of a theory $T$
along a signature morphism $\sigma : \mathit{sig}(T)\to\Sigma$,
resulting in a structured theory $\sigma(T)$ which has signature
$\Sigma$, while the model class is $\{ M \in \instMod{\SM}(\Sigma)
\mid M\reductop_\sigma \in \instMod{\SM}(T) \}$, i.e., models are
those $\Sigma$-models that reduce (via $\sigma$) to a $T$-model.
Moreover, sometimes we want to \emph{drop events} (like \uml{card} in
Fig.~\ref{fig:state-machine}) when moving to a more abstract state
machine.  This can be modelled by a notion dual to translation, namely
\emph{hiding}. Given a theory $T$ and a signature morphism $\theta:
\Sigma\to \mathit{sig}(T)$, the structured theory $\theta^{-1}(T)$ has
signature $\Sigma$, while the model class is $\{ M\reductop_\theta \in
\instMod{\SM}(\Sigma) \mid M \in \instMod{\SM}(T) \}$, i.e., models
are all $\theta$-reducts of $T$-models.  Altogether, we arrive at

\begin{definition}
An ``abstract'' (behavioural or protocol) state machine $T_1$
\emph{refines into} a ``concrete'' state machine $T_2$ via signature
morphisms $\theta: Sig(T_1) \to \Sigma$ and $\sigma :
\mathit{sig}(T_2) \to \Sigma$ into some ``mediating signature''
$\Sigma$, if
\begin{equation*}
  T_1 \leadsto \theta^{-1}(\sigma(T_2))
\end{equation*}
in other words, for all $\Sigma$-models $M$
\begin{equation*}
  M\reductop_\sigma \in \instMod{\SM}(T_2)
\Rightarrow
  M\reductop_\theta \in \instMod{\SM}(T_1)
\end{equation*}
\end{definition}

Concerning our original refinement question stated in
Ex.~\ref{ex:ref}, we now can argue: As the state machine of the
\uml{atm}, shown in Fig.~\ref{fig:state-machine} is a refinement of
the protocol state machine in Fig.~\ref{fig:psm}, using a suitable
signature morphism, the interleaving product of the \uml{atm} and
\uml{bank} state machine, in the syntactic version, will be so as
well. As furthermore the protocol state machine has no post
conditions, we have established a refinement.


\section{Conclusions}
\label{sec:conclusions}

We have presented institutions for behavioural and protocol UML state
machines and have studied an interleaving product and a notion of
determinism. We furthermore presented first steps of how to study
refinement in such a context. Our institutions provide the necessary
prerequisites for including UML state machines into a heterogeneous
institution-based UML semantics and to develop their relationship to
other UML sub-languages and diagram types.

An important future extension for the state machine institutions is to
add hierarchical states, and to consider refinements from hierarchical
to flat state machines.  For an integration into the software
development process, the study of correct code generation is
indispensable.  The Heterogeneous Tool Set
(\protect\Hets~\cite{MossakowskiEtAl06,MossakowskiEA06}) provides
analysis and proof support for multi-logic specifications, based on a
strong semantic (institution-based) backbone. Implementation of proof
support for UML state machines (and other kinds of UML diagrams) is
under way.
 
\bibliographystyle{splncs} 
\bibliography{knapp,uml,acsl,inst,hets,uml-inst}

\begin{thebibliography}{10}

\bibitem{uml-2.4.1-superstructure}
{Object Management Group}:
\newblock {Unified Modeling Language}.
\newblock Standard formal/2011-08-06, OMG (2011)

\bibitem{knapp-mossakowski-roggenbach:corr:2014}
Knapp, A., Mossakowski, T., Roggenbach, M.:
\newblock {An Institutional Framework for Heterogeneous Formal Development in
  UML} (2014) \url{http://arxiv.org/abs/1403.7747}.

\bibitem{GoguenBurstall92}
Goguen, J.A., Burstall, R.M.:
\newblock Institutions: Abstract model theory for specification and
  programming.
\newblock J. ACM \textbf{39} (1992)  95--146

\bibitem{grosse-rhode:2004}
Gro{\ss}e-Rhode, M.:
\newblock Semantic Integration of Heterogeneous Software Specifications.
\newblock Monographs in Theoretical Computer Science. Springer (2004)

\bibitem{roggenbach06}
Roggenbach, M.:
\newblock {CSP-CASL: A New Integration of Process Algebra and Algebraic
  Specification}.
\newblock Theo. Comp. Sci. \textbf{354} (2006)  42--71

\bibitem{MossakowskiRoggenbach07}
Mossakowski, T., Roggenbach, M.:
\newblock {Structured CSP --- A Process Algebra as an Institution}.
\newblock In: WADT 2006. LNCS 4409.
\newblock Springer (2007)  92--110

\bibitem{oreilly12b}
O'Reilly, L., Mossakowski, T., Roggenbach, M.:
\newblock {Compositional Modelling and Reasoning in an Institution for
  Processes and Data}.
\newblock In: WADT 2010. LNCS 7137, Springer (2012)  251--269

\bibitem{schattkowsky-mueller:vlhcc:2005}
Schattkowsky, T., Müller, W.:
\newblock {Transformation of UML State Machines for Direct Execution}.
\newblock In: VL/HCC 2005, IEEE (2005)  117--124

\bibitem{fecher-schoenborn:fmics:2007}
Fecher, H., Schönborn, J.:
\newblock {UML 2.0 State Machines: Complete Formal Semantics via Core State
  Machines}.
\newblock In: FMICS/PDMC 2006. LNCS 4346, Springer (2007)  244--260

\bibitem{calegari-szasz:isse:2011}
Calegari, D., Szasz, N.:
\newblock {Institutionalising UML 2.0 State Machines}.
\newblock Innov. Syst. Softw. Eng. \textbf{7} (2011)  315--323

\bibitem{pscs-beta-1}
{Object Management Group}:
\newblock {Precise Semantics of UML Composite Structures}.
\newblock Beta Specification ptc/14-06-15, OMG (2014)

\bibitem{derezinska-szczykulski:cst:2012}
Derezińska, A., Szczykulski, M.:
\newblock {Interpretation Problems in Code Generation from UML State Machines
  --- A Comparative Study}.
\newblock In Kwater, T., ed.: Computing in Science and Technology 2011:
  Monographs in Applied Informatics.
\newblock Warsaw University (2012)  36--50

\bibitem{lano:uml2:2009}
Lano, K., ed.:
\newblock UML~2 --- Semantics and Applications.
\newblock Wiley (2009)

\bibitem{DBLP:conf/uml/EngelsHK03}
Engels, G., Heckel, R., Küster, J.M.:
\newblock {The Consistency Workbench: A Tool for Consistency Management in
  {UML}-Based Development}.
\newblock In: UML'03. LNCS 2863, Springer (2003)

\bibitem{caslref}
Mossakowski, T., Sannella, D., Tarlecki, A.:
\newblock {A Simple Refinement Language for \textsc{Casl}}.
\newblock In: WADT 2004. LNCS 3423.
\newblock Springer (2005)

\bibitem{CASL-refinement-journal}
Codescu, M., Mossakowski, T., Sannella, D., Tarlecki, A.:
\newblock {Specification Refinements: Calculi, Tools, and Applications} (2014)
  Submitted.

\bibitem{Borzyszkowski:2002:LSS:639756.639759}
Borzyszkowski, T.:
\newblock Logical systems for structured specifications.
\newblock Theor. Comput. Sci. \textbf{286} (2002)  197--245

\bibitem{MossakowskiEtAl06}
Mossakowski, T., Autexier, S., Hutter, D.:
\newblock {Development Graphs --- Proof Management for Structured
  Specifications}.
\newblock J. Log. Alg. Program. \textbf{67} (2006)  114--145

\bibitem{Diaconescu02}
Diaconescu, R.:
\newblock {Grothendieck Institutions}.
\newblock Applied Cat. Struct. \textbf{10} (2002)  383--402

\bibitem{MossakowskiEA06}
Mossakowski, T., Maeder, C., Lüttich, K.:
\newblock {The Heterogeneous Tool Set}.
\newblock In: TACAS'07. LNCS 4424, Springer (2007)

\end{thebibliography}

\newpage
\appendix
\section{Proofs of the theorems}

\begin{proposition}
The action institution over a given institution of guards enjoys the
satisfaction condition.
\end{proposition}
\begin{proof}
\begin{equation*}
\begin{array}{ll}
&\Omega'\reductop_\eta\models g[a]\overline{m}\rhd g'\\
\mbox{iff}&\forall\omega_1,\omega_2:V_{\Eta_1}\to\Val\,\forall\overline{m}'\,. (\omega_1\models g\wedge\shorttrans{\omega_1}{a, \overline{m}'}{\Omega'\reductop_\eta}{\omega_2})\Rightarrow \omega_2\models g'\wedge\overline{m}\subseteq\overline{m}'\\
\mbox{iff}&\forall\omega_1,\omega_2:V_{\Eta_1}\to\Val\,\forall\overline{m}'\,. (\omega_1\models g\wedge\shorttrans{\omega_1}{a, \eta_M^{-1}(\overline{m}')}{\Omega'\reductop_\eta}{\omega_2})\Rightarrow \omega_2\models g'\wedge\overline{m}\subseteq\eta_M^{-1}(\overline{m}')\\
\mbox{iff}&\forall\omega'_1,\omega'_2:V_{\Eta_2}\to\Val\,\forall\overline{m}'\,. (\omega'_1\reductop_\eta\models g\wedge\shorttrans{\omega'_1}{\eta_A(a), \overline{m}'}{\Omega'}{\omega'_2})\Rightarrow \omega'_2\reductop_\eta\models g'\wedge\eta_M(\overline{m})\subseteq\overline{m}'\\
\mbox{iff}&\forall\omega'_1,\omega'_2:V_{\Eta_2}\to\Val\,\forall\overline{m}'\,. (\omega'_1\models G(\eta_V)(g)\wedge\shorttrans{\omega'_1}{\eta_A(a), \overline{m}'}{\Omega'}{\omega'_2})\Rightarrow \omega'_2\models G(\eta_V)(g')\wedge\eta_M(\overline{m})\subseteq\overline{m}'\\
\mbox{iff}&
\Omega'\models G(\eta_V)(g)[\eta_A(a)]\eta_M(\overline{m}) \rhd
G(\eta_V)(g')\\
\mbox{iff}&
\Omega'\models \eta(g[a]\overline{m}\rhd g')
\end{array}
\end{equation*}
The third last step uses the satisfaction condition of the institution
of guards.
\end{proof}

\begin{proposition}
The state machine institution over given institutions of actions and
guards enjoys the satisfaction condition.
\end{proposition}
\begin{proof}
We only consider the first case of the condition for $\Delta_\Theta$ in
the definition of the satisfaction relation; the second case is similar.
Using preservation
of internal messages, one can see that
\begin{gather*}
  \shorttrans{(\omega, p :: \sigma_P^{-1}(\overline{p}), s)}{\overline{m} \setminus E_{\Sigma}}{\Delta_{\Theta}}{(\omega', \sigma_P^{-1}(\overline{p}) \lhd ((\overline{m} \cap E_{\Sigma}) \cup \overline{f}), s')}
\quad\text{if}\\\qquad
  \exists \shorttrans{s}{\sigma_P^{-1}(p)[g]/a, \overline{f}}{T}{s'} \,.\, \omega \models g \land \shorttrans{\omega}{a,\overline{m}}{\Omega}{\omega'}
\end{gather*}
is equivalent to
\begin{gather*}
  \shorttrans{(\omega, \sigma_P(p) :: \overline{p}, \sigma_P(s))}{\overline{m} \setminus E_{\Sigma'}}{\Delta_{\Theta}}{(\omega', \overline{p} \lhd ((\overline{m} \cap E_{\Sigma'}) \cup \sigma_F(\overline{f})), \sigma_P(s'))}
\quad\text{if}\\\qquad
  \exists \shorttrans{\sigma_S(s)}{\sigma_P(p)[g]/a, \sigma_F(\overline{f})}{\sigma(T)}{\sigma_P(s')} \,.\, \omega \models g \land \shorttrans{\omega}{a,\overline{m}}{\Omega}{\omega'}
\end{gather*}

\end{proof}

\ \\
\textbf{Theorem \ref{thm:compatible-interleaving}.}
 \textit{If  $\langle\Omega_i, \Theta_i\rangle \models (s_0^i, T_i)$ ($i=1,2$),
then $\langle\Omega_1, \Theta_1\rangle \pll \langle\Omega_2,
  \Theta_2\rangle\models (s_0^1, T_1) \pll (s_0^2,T_2)$.}
\begin{proof}
For simplicity, we concentrate on the first condition of $\Delta_\Theta$
in the satisfaction condition.
By construction of $\Delta_{\Theta_1}\pll\Delta_{\Theta_2}$,
\begin{gather*}
  \longtrans{(\omega, p :: (\overline{p}_1 \cup \overline{p}_2), (s_1, s_2))}
            {\overline{m} \setminus E_{\hat{\Sigma}}}
            {\Delta_{\Theta_1} \pll \Delta_{\Theta_2}}
            {(\omega', (\overline{p}_1 \cup \overline{p}_2) \lhd ((\overline{p}' \cup \overline{m}) \cap E_{\hat{\Sigma}}), (s_1', s_2'))}
\\\qquad
\begin{array}{@{}ll@{}}
\text{iff } & \exists i \in \{ 1, 2 \} \,.\, \shorttrans{(\omega\reductop_{V_{\Eta_i}}, p :: \overline{p}_i, s_i)}{  \overline{m}}{\Delta_{\Theta_i}}{(\omega'\reductop_{V_{\Eta_i}}, \overline{p}_i \lhd \overline{p}', s_i')} \land{}\\
            &\qquad \forall j \in \{ 1, 2 \} \setminus \{ i \} \,.\, (\omega\reductop_{V_{\Eta_j}}, \overline{p}_j, s_j)=(\omega'\reductop_{V_{\Eta_j}}, \overline{p}_j, s_j')
\ \text{.}
\end{array}
\end{gather*}
By the assumption $\langle\Omega_i, \Theta_i\rangle \models (s_0^i, T_i)$ ($i=1,2$), the latter condition is equivalent to
\begin{gather*}
\begin{array}{@{}ll@{}}
 & \exists \overline{m}'\,.  \overline{m}'\setminus E_\Sigma=\overline{m}\land (\overline{m}'\cap E_\Sigma)\cup\overline{f}=\overline{p}'\land \exists i \in \{ 1, 2 \} \,.\,  \\
&\qquad\exists \shorttrans{s_i}{p[g]/a, \overline{f}}{T_i}{s'_i} \,.\, \omega \models g \land \shorttrans{\omega\reductop_{V_{\Eta_i}}}{a,\overline{m}'}{\Omega}{\omega'\reductop_{V_{\Eta_i}}} \land{}\\
            &\qquad \forall j \in \{ 1, 2 \} \setminus \{ i \} \,.\, \omega\reductop_{V_{\Eta_j}}=\omega'\reductop_{V_{\Eta_j}}\land s_j=s_j'
\ \text{.}
\end{array}
\end{gather*}
But this in turn is equivalent to 
\begin{gather*}
\begin{array}{l}
\exists \overline{m}'\,.  \overline{m}'\setminus E_\Sigma=\overline{m}\land (\overline{m}'\cap E_\Sigma)\cup\overline{f}=\overline{p}'\land \\
\exists i \in \{ 1, 2 \} \,.\, \shorttrans{s_i}{p[g]/a, \overline{f}}{T_i}{s'_i}
\land\omega \models g \land
  \forall j \in \{ 1, 2 \} \setminus \{ i \} \,.\, s_j = s'_j
\land 
\shorttrans{\omega}{a,\overline{m}'}{\Omega}{\omega'}.
\end{array}
\end{gather*}
i.e.
\begin{gather*}
\begin{array}{l}
\exists \overline{m}'\,.  \overline{m}'\setminus E_\Sigma=\overline{m}\land (\overline{m}'\cap E_\Sigma)\cup\overline{f}=\overline{p}'\land \\
\shorttrans{(s_1, s_2)}{p[g]/a, \overline{f}}{T_1\pll T_2}{(s'_1, s'_2)}
\land \omega \models g \land
\shorttrans{\omega}{a,\overline{m}'}{\Omega}{\omega'}.
\end{array}
\end{gather*}

\newcommand{\eat}[1]{}
Altogether, we get
\begin{gather*}
  \longtrans{(\omega, p :: (\overline{p}_1 \cup \overline{p}_2), (s_1, s_2))}
            {\overline{m} \setminus E_{\hat{\Sigma}}}
            {\Delta_{\Theta_1} \pll \Delta_{\Theta_2}}
            {(\omega', (\overline{p}_1 \cup \overline{p}_2) \lhd ((\overline{p}' \cup \overline{m}) \cap E_{\hat{\Sigma}}), (s_1', s_2'))}
\\\qquad
\begin{array}{@{}ll@{}}
\text{iff } & \exists \overline{m}'\,.  \overline{m}'\setminus E_\Sigma=\overline{m}\land (\overline{m}'\cap E_\Sigma)\cup\overline{f}=\overline{p}'\land \\
&\shorttrans{(s_1, s_2)}{p[g]/a, \overline{f}}{T_1\pll T_2}{(s'_1, s'_2)}
\land \omega \models g \land
\shorttrans{\omega}{a,\overline{m}'}{\Omega}{\omega'}.
\end{array}
\end{gather*}
which by taking $\overline{p}=\overline{p_1}\cup\overline{p_1}$ amounts to
\begin{gather*}
  \longtrans{(\omega, p :: \overline{p}, (s_1, s_2))}
            {\overline{m}'\setminus E_\Sigma}
            {\Delta_{\Theta_1} \pll \Delta_{\Theta_2}}
            {(\omega', \overline{p} \lhd (\overline{m}'\cap E_\Sigma)\cup\overline{f}, (s_1', s_2'))}
\\\qquad
\begin{array}{@{}ll@{}}
\text{iff } & \exists 
\shorttrans{(s_1, s_2)}{p[g]/a, \overline{f}}{T_1\pll T_2}{(s'_1, s'_2)}
\land \omega \models g \land
\shorttrans{\omega}{a,\overline{m}'}{\Omega}{\omega'}.
\end{array}
\end{gather*}
which means
$$\langle\Omega_1, \Theta_1\rangle \pll \langle\Omega_2,
  \Theta_2\rangle\models (s_0^1, T_1) \pll (s_0^2,T_2)$$
\end{proof}

\ \\
\textbf{Theorem \ref{th:monoid}.} \textit{The set of state machine structures (over all signatures) with
interleaving product $\pll$ forms a discrete symmetric monoidal
category, which is a ``commutative monoid up to isomorphism''.
} \\
\textit{To show: $(\Omega_1 \pll \Omega_2) \pll \Omega_3 \cong \Omega_1 \pll (\Omega_2 \pll \Omega_3)$}\\
The associativity follows from \\
$\shorttrans{\omega}{a,\overline{m}}{(\Omega_1 \pll \Omega_2) \pll \Omega_3}{\omega'}$ 
\[\begin{array}{rll}
\text{iff} & & a \in A_{\Eta_1} \cup A_{\Eta_2}, \shorttrans{\omega\reductop_{V_{\Eta_1} \cup V_{\Eta_2}}}{a,\overline{m}}{\Omega_1 \pll \Omega_2}{\omega'\reductop_{V_{\Eta_1} \cup V_{\Eta_2}}}, \omega\reductop_{V_{\Eta_3} \setminus (V_{\Eta_1} \cup V_{\Eta_2})} = \omega'\reductop_{V_{\Eta_3} \setminus (V_{\Eta_1} \cup V_{\Eta_2})} \\
& \vee &   a \in A_{\Eta_3}, \shorttrans{\omega\reductop_{V_{\Eta_3}}}{a,\overline{m}}{\Omega_3}{\omega'\reductop_{V_{\Eta_3}}}, \omega\reductop_{(V_{\Eta_1} \cup V_{\Eta_2}) \setminus V_{\Eta_3}} = \omega'\reductop_{(V_{\Eta_1} \cup V_{\Eta_2}) \setminus V_{\Eta_3} } \\

\text{iff} & & a \in A_{\Eta_1}, \shorttrans{\omega\reductop_{V_{\Eta_1}}}{a,\overline{m}}{\Omega_1}{\omega'\reductop_{V_{\Eta_1}}}, \omega\reductop_{V_{\Eta_3} \setminus V_{\Eta_1}} = \omega'\reductop_{V_{\Eta_3} \setminus V_{\Eta_1}}, \omega\reductop_{V_{\Eta_2} \setminus V_{\Eta_1}} = \omega'\reductop_{V_{\Eta_2} \setminus V_{\Eta_1}} \\
& \vee & a \in A_{\Eta_2}, \shorttrans{\omega\reductop_{V_{\Eta_2}}}{a,\overline{m}}{\Omega_2}{\omega'\reductop_{V_{\Eta_2}}}, \omega\reductop_{V_{\Eta_3} \setminus V_{\Eta_2}} = \omega'\reductop_{V_{\Eta_1} \setminus V_{\Eta_2}}, \omega\reductop_{V_{\Eta_1} \setminus V_{\Eta_2}} = \omega'\reductop_{V_{\Eta_3} \setminus V_{\Eta_2}} \\
& \vee &   a \in A_{\Eta_3}, \shorttrans{\omega\reductop_{V_{\Eta_3}}}{a,\overline{m}}{\Omega_3}{\omega'\reductop_{V_{\Eta_3}}}, \omega\reductop_{V_{\Eta_1} \setminus V_{\Eta_3}} = \omega'\reductop_{V_{\Eta_1} \setminus V_{\Eta_3}}, \omega\reductop_{V_{\Eta_2} \setminus V_{\Eta_3}} = \omega'\reductop_{V_{\Eta_2} \setminus V_{\Eta_3}} \\

\text{iff} & & a \in A_{\Eta_2} \cup A_{\Eta_3}, \shorttrans{\omega\reductop_{V_{\Eta_2} \cup V_{\Eta_3}}}{a,\overline{m}}{\Omega_2 \pll \Omega_3}{\omega'\reductop_{V_{\Eta_2} \cup V_{\Eta_3}}}, \omega\reductop_{V_{\Eta_1} \setminus (V_{\Eta_2} \cup V_{\Eta_3})} = \omega'\reductop_{V_{\Eta_1} \setminus (V_{\Eta_2} \cup V_{\Eta_3})} \\
& \vee &   a \in A_{\Eta_1}, \shorttrans{\omega\reductop_{V_{\Eta_1}}}{a,\overline{m}}{\Omega_1}{\omega'\reductop_{V_{\Eta_1}}}, \omega\reductop_{(V_{\Eta_2} \cup V_{\Eta_3}) \setminus V_{\Eta_1}} = \omega'\reductop_{(V_{\Eta_2} \cup V_{\Eta_3}) \setminus V_{\Eta_1} } \\

\text{iff} & & \shorttrans{\omega}{a,\overline{m}}{\Omega_1 \pll (\Omega_2 \pll \Omega_3)}{\omega'}
\end{array}\]
\textit{To show: $\Omega_1 \pll \Omega_2 \cong \Omega_2 \pll \Omega_1$}\\
The commutativity follows directly from the definition of the product of action relations. \\ \\

\begin{proof}
\textit{To show:} $(\Delta_{\Omega_1} \pll \Delta_{\Omega_2}) \pll  \Delta_{\Omega_3} \cong  \Delta_{\Omega_1} \pll ( \Delta_{\Omega_2} \pll  \Delta_{\Omega_3})$\\
Note that the event sets are disjoint, i.e. each event triggers a step in at most one state machine. Thus the existence condition in the definition of $\Delta_{\Theta_1} || \Delta_{\Theta_2}$ can be written as a disjunction of all possible cases:
\[
\begin{array}{rll}
\multicolumn{2}{l}{\shorttrans{(\omega, e :: (\overline{e_1} \cup \overline{e_2}), (s_1,s_2))}{\overline{m}  \backslash E_{\hat{\Sigma}}}{\Delta_{\Theta_1} || \Delta_{\Theta_2}}{(\omega', (\overline{e_1} \cup \overline{e_2})  \lhd ((\overline{e}' \cup \overline{m}) \cap E_{\hat{\Sigma}}), (s_1',s_2'))}} \\ 
\text{iff} 
&  \\
& ( \shorttrans{(\omega\reductop_{V_{\Eta_1}}, e :: \overline{e_1},s_1)}{\overline{m}}{\Delta_{\Theta_1}}{(\omega'\reductop_{V_{\Eta_1}} \wedge \overline{e_1} \lhd \overline{e}',s_1')}, \\ 
& \qquad \emptyTrans{2})  \\
& \vee  (\shorttrans{(\omega\reductop_{V_{\Eta_2}}, e :: \overline{e_2},s_2)}{\overline{m}}{\Delta_{\Theta_2}}{(\omega'\reductop_{V_{\Eta_2}}, \overline{e_2} \lhd \overline{e}',s_2')}\wedge \\ 
& \qquad  \emptyTrans{1}) \\
\end{array}
\]
If a set of messages $\overline{m} \backslash E_{\hat{\Sigma}}$ is emitted by the machine represented by the left interleaving product $\Delta_{\Theta_1} \pll  \Delta_{\Theta_2}$ it is possible that it contains elements from $E_{\Sigma_3}$ or $M_{\Eta_3}$. However it is not possible that any events of $E_{\hat{\Sigma}}$ are contained. In order to construct the messages emitted by the machine represented by $(\Delta_{\Theta_1} \pll  \Delta_{\Theta_2}) \pll  \Delta_{\Theta_3}$ the messages are stripped off any events contained in $E_{\Sigma_3}$:
\[\begin{array}{cl}

(\overline{m} \backslash (E_{\Sigma_1} \cup E_{\Sigma_2})) \backslash ((E_{\Sigma_1} \cup E_{\Sigma_2})  \cup E_{\Sigma_3})
\end{array}\]
The previous filtering of messages by the left interleaving product becomes obsolete by the later applied, more granulated one, done by the whole product. So it is sufficient to regard the latter filtering solely:
\[\begin{array}{cl}

(\overline{m} \backslash (E_{\Sigma_1} \cup E_{\Sigma_2})) \backslash ((E_{\Sigma_1} \cup E_{\Sigma_2})  \cup E_{\Sigma_3})
=  \overline{m}  \backslash ((E_{\Sigma_1} \cup E_{\Sigma_2})  \cup E_{\Sigma_3})
\end{array}\]

Same applies to the messages emitted by the outer machine. Using these results the associativity results from
\[
\begin{array}{lll}
\multicolumn{2}{l}{\llongtrans{ (\omega, e :: (\overline{e} \cup \overline{e_3}), (s,s_3)) }{
\overline{m} \backslash ((E_{\Sigma_1} \cup E_{\Sigma_2})  \cup E_{\Sigma_3})}{(\Delta_{\Theta_1} \pll  \Delta_{\Theta_2}) \pll  \Delta_{\Theta_3}}{}}\\
\multicolumn{2}{l}{\qquad(\omega', (\overline{e} \cup \overline{e_3})  \lhd ((\overline{e}' \cup \overline{m}) \cap ((E_{\Sigma_1} \cup E_{\Sigma_2})  \cup E_{\Sigma_3})), (s',s_3'))} \\ 
\text{iff } & \\
& \begin{array}{ll}
& (\shorttrans{(\omega\reductop_{V_{\Eta_1} \cup V_{\Eta_2}}, e :: \overline{e},s)}{\overline{m}}{\Delta_{\Theta_1} \pll  \Delta_{\Theta_2}}{(\omega', \overline{e_i} \lhd \overline{e}',s')} \\ 
& \wedge \emptyTrans{3}) \\ 
 \vee & \\ 
 & ((\omega\reductop_{V_{\Eta_1} \cup V_{\Eta_2}}, \overline{e} ,s)=(\omega'\reductop_{V_{\Eta_1} \cup V_{\Eta_2}}, \overline{e},s') \\ 
 & \qquad \in (V_{\Eta_1} \cup V_{\Eta_2} \rightarrow \Val) \times (E_{\Eta_1} \cup E_{\Eta_2}) \times (S_{\Eta_1} \cup S_{\Eta_2}) \\
&  \wedge \shorttrans{(\omega\reductop_{V_{\Eta_3}}, e :: \overline{e_3} ,s_3)}{\overline{m}}{\Delta_{\Theta_3}}{(\omega'\reductop_{V_{\Eta_3}}, \overline{e_3}  \lhd \overline{e}',s_3')}) \\

\end{array} 
\end{array}\]
This can be rewritten as
\begin{equation}
\label{eq:Transaction}
\begin{array}{l}
\llongtrans{((\omega, e :: ( (\overline{e_1} \cup \overline{e_2}) \cup \overline{e_3}), ((s_1,s_2),s_3))}{\overline{m}   \backslash ((E_{\Sigma_1} \cup E_{\Sigma_2})  \cup E_{\Sigma_3})}{(\Delta_{\Theta_1} \pll  \Delta_{\Theta_2}) \pll  \Delta_{\Theta_3}}{} \\
 { ((\omega', ((\overline{e_1} \cup \overline{e_2}) \cup \overline{e_3})  \lhd ((\overline{e}' \cup \overline{m}) \cap ((E_{\Sigma_1} \cup E_{\Sigma_2})  \cup E_{\Sigma_3})), ((s_1',s_2'),s_3'))}
 \end{array}
\end{equation}

which results in a slightly different condition
\[
\begin{array}{lll}
\text{iff} &  & (\shorttrans{(\omega\reductop_{V_{\Eta_1} \cup V_{\Eta_2}}, e :: (\overline{e_1} \cup \overline{e_2}), (s_1,s_2))}{\overline{m}}{\Delta_{\Theta_1} || \Delta_{\Theta_2}}{} \\ 
& & \qquad (\omega'\reductop_{V_{\Eta_1} \cup V_{\Eta_2}}, (\overline{e_1} \cup \overline{e_2})  \lhd ((\overline{e}' \cup \overline{m}) \cap ((E_{\Sigma_1} \cup E_{\Sigma_2})  \cup E_{\Sigma_3})), (s_1',s_2')) \\ 
&  &\wedge \emptyTrans{3})  \\ 
& \vee &((\omega\reductop_{V_{\Eta_1} \cup V_{\Eta_2}},  (\overline{e_1} \cup \overline{e_2}), (s_1,s_2)) = (\omega'\reductop_{V_{\Eta_1} \cup V_{\Eta_2}}, \overline{e_1} \cup \overline{e_2}, (s_1',s_2')) \\
& & \qquad \in ((V_{\Eta_1} \cup V_{\Eta_2}) \rightarrow   \Val) \times \powerset(E_{\Sigma_1} \cup E_{\Sigma_2}) \times (S_{\Sigma_1} \cup S_{\Sigma_2}) \\ 
& & \wedge \shorttrans{(\omega\reductop_{V_{\Eta_3}}, e ::\overline{e_3} ,s_3)}{\overline{m}}{\Delta_{\Theta_3}}{(\omega'\reductop_{V_{\Eta_3}}, \overline{e_3} \lhd \overline{e}',s_3'))}  \\
\end{array}
\]
\[
\begin{array}{rll}
\text{iff} &  \\
& &(\qquad((\shorttrans{(\omega\reductop_{V_{\Eta_1}}, e :: \overline{e_1},s_1)}{\overline{m}}{\Delta_{\Theta_1}}{(\omega'\reductop_{V_{\Eta_3}}, \overline{e_1} \lhd \overline{e}',s_1')} \\
& & \qquad \wedge \emptyTrans{2}) \\
& & \vee \\
& & \qquad(\emptyTrans{1} \\ 
& & \qquad\wedge  \shorttrans{(\omega\reductop_{V_{\Eta_2}}, e :: \overline{e_2},s_2)}{\overline{m}}{\Delta_{\Theta_2}}{(\omega'\reductop_{V_{\Eta_3}}, \overline{e_2} \lhd \overline{e}' ,s_2')})) \\ 
& & \wedge \emptyTrans{3}) \\ 
& \vee \\
 & & ( \emptyTrans{1} \\ 
& & \wedge \emptyTrans{2} \\
& & \wedge \shorttrans{(\omega\reductop_{V_{\Eta_3}}, e ::\overline{e_3} ,s_3)}{\overline{m}}{\Delta_{\Theta_3}}{(\omega'\reductop_{V_{\Eta_3}}, \overline{e_3} \lhd \overline{e}',s_3'))} 
\end{array}
\]
\[\begin{array}{rll}
\text{iff} & & (\emptyTrans{1}\\
 & & \wedge( \qquad (\shorttrans{(\omega\reductop_{V_{\Eta_2}}, e :: \overline{e_2},s_2)}{\overline{m}}{\Delta_{\Theta_2}}{(\omega'\reductop_{V_{\Eta_2}}, \overline{e_2} \lhd \overline{e}',s_2')} \\ 
& &  \qquad \wedge \emptyTrans{3}) \\
& &\quad \vee \\
& &\qquad  (\emptyTrans{2} \\
& &\qquad  \wedge \shorttrans{(\omega\reductop_{V_{\Eta_3}}, e :: \overline{e_3},s_3)}{\overline{m}}{\Delta_{\Theta_3}}{(\omega'\reductop_{V_{\Eta_3}}, \overline{e_3} \lhd \overline{e}' ,s_3')}))\\  
& & \hphantom{\wedge}) \\ 
&  \vee & \\
& &  (\shorttrans{(\omega\reductop_{V_{\Eta_1}}, e ::\overline{e_1} ,s_1)}{\overline{m}}{\Delta_{\Theta_1}}{(\omega'\reductop_{V_{\Eta_3}}, \overline{e_1} \lhd \overline{e}',s_1')} \\
& &\wedge \emptyTrans{2} \\
& &\wedge \emptyTrans{3})\\
\end{array}
\]
This equivalence reveals that the transitions condition is associative. This result and the isomorphism of transition (\ref{eq:Transaction}) to 
\[
\begin{array}{l}
 \llongtrans{(\omega, e :: (\overline{e_1} \cup (\overline{e_2} \cup \overline{e_3})), (s_1,(s_2,s_3)))}{\overline{m}   \backslash (E_{\Sigma_1} \cup (E_{\Sigma_2}  \cup E_{\Sigma_3})}{\Delta_{\Theta_1} \pll  (\Delta_{\Theta_2} \pll  \Delta_{\Theta_3})}{} \\
 { (\omega', (\overline{e_1} \cup (\overline{e_2} \cup \overline{e_3}))  \lhd ((\overline{e}' \cup \overline{m}) \cap (E_{\Sigma_1} \cup (E_{\Sigma_2}  \cup E_{\Sigma_3}))), (s_1',(s_2',s_3')))}
\end{array}
\]
concludes the proof of associativity.\\ 
\\
\textit{To show: $\Omega_1 \pll \Omega_2 \cong \Omega_2 \pll \Omega_1$} 

\[
\begin{array}{rll}
\multicolumn{3}{l}{\shorttrans{(\omega, e :: (\overline{e_1} \cup \overline{e_2}), (s_1,s_2))}{\overline{m}  \backslash E_{\hat{\Sigma}}}{\Delta_{\Theta_1} || \Delta_{\Theta_2}}{(\omega', (\overline{e_1} \cup \overline{e_2})  \lhd ((\overline{e}' \cup \overline{m}) \cap E_{\hat{\Sigma}}), (s_1',s_2'))}} \\ 
\text{iff} 
& & \\
& &(\shorttrans{(\omega\reductop_{V_{\Eta_1}}, e :: \overline{e_1},s_1)}{\overline{m}}{\Delta_{\Theta_1}}{(\omega'\reductop_{V_{\Eta_1}}, \overline{e_1} \lhd \overline{e}',s_1')} \\ 
& & \wedge \emptyTrans{2})  \\
& \vee &  \\
& &(\shorttrans{(\omega\reductop_{V_{\Eta_2}}, e :: \overline{e_2},s_2)}{\overline{m}}{\Delta_{\Theta_2}}{(\omega'\reductop_{V_{\Eta_2}}, \overline{e_2} \lhd \overline{e}',s_2')} \\ 
& &\wedge  \emptyTrans{1})\\

\text{iff} 
& & \\
& &(\shorttrans{(\omega\reductop_{V_{\Eta_2}}, e :: \overline{e_2},s_2)}{\overline{m}}{\Delta_{\Theta_2}}{(\omega'\reductop_{V_{\Eta_2}}, \overline{e_2} \lhd \overline{e}',s_2')}, \\ 
& & \wedge  \emptyTrans{1})  \\
& \vee  & \\
& & (\shorttrans{(\omega\reductop_{V_{\Eta_1}}, e :: \overline{e_1},s_1)}{\overline{m}}{\Delta_{\Theta_1}}{(\omega'\reductop_{V_{\Eta_1}}, \overline{e_1} \lhd \overline{e}',s_1')} \\ 
& & \wedge \emptyTrans{2})  \\

\end{array}
\]
The last part of this equivalence is also the condition of the respective transition in $\Delta_{\Theta_2} \pll \Delta_{\Theta_1}$ Furthermore are left and the right hand side of  
\[
\shorttrans{(\omega, e :: (\overline{e_1} \cup \overline{e_2}), (s_1,s_2))}{\overline{m}  \backslash E_{\hat{\Sigma}}}{\Delta_{\Theta_1} || \Delta_{\Theta_2}}{(\omega', (\overline{e_1} \cup \overline{e_2})  \lhd ((\overline{e}' \cup \overline{m}) \cap E_{\hat{\Sigma}}), (s_1',s_2'))}
\]
isomorphic to 
\[
(\omega, e :: (\overline{e_2} \cup \overline{e_1}), (s_2,s_1)) \text{ and } (\omega', (\overline{e_2} \cup \overline{e_1})  \lhd ((\overline{e}' \cup \overline{m}) \cap E_{\hat{\Sigma}}), (s_2',s_1'))
\]
which is also a transition in $\Delta_{\Theta_2} \pll \Delta_{\Theta_1}$.
\\
\\

\textit{To show: $\exists (\Omega_\varepsilon,\Theta_\epsilon): \Delta_{\Theta_1} \pll \Delta_{\Theta_\varepsilon} \cong \Delta_{\Theta}$ and $\Omega_1 \pll \Omega_\varepsilon \cong \Omega_1$ } \\
Consider the state machine 
\begin{enumerate}
\item $E_{\Sigma_\varepsilon}= \emptyset$
\item $F_{\Sigma_\varepsilon}= \emptyset$
\item $S_{\Sigma_\varepsilon} = \{s_0\} \not \subseteq S_{\Sigma_1}$
\item $\Eta_\varepsilon = (\emptyset,\emptyset,\emptyset)$
\item $\Theta_\varepsilon = ((f:\emptyset \rightarrow \emptyset, s_0),\emptyset)$
\end{enumerate}
The resulting transition relation delivers:
\[
\begin{array}{rll}
\multicolumn{3}{l}{\shorttrans{(\omega, e :: (\overline{e_1} \cup \overline{e_\varepsilon}), (s_1,s_\varepsilon))}{\overline{m}  \backslash E_{\hat{\Sigma}}}{\Delta_{\Theta_1} || \Delta_{\Theta_\varepsilon}}{(\omega', (\overline{e_1} \cup \overline{e_\varepsilon})  \lhd ((\overline{e}' \cup \overline{m}) \cap E_{\hat{\Sigma}}), (s_1',s_\varepsilon'))}} \\ 
\text{iff} 
&  \\
& &(\shorttrans{(\omega\reductop_{V_{\Eta_1}}, e :: \overline{e_1},s_1)}{\overline{m}}{\Delta_{\Theta_1}}{(\omega'\reductop_{V_{\Eta_1}}, \overline{e_1} \lhd \overline{e}',s_1')} \\ 
& &\wedge \emptyTrans{\varepsilon})  \\
& \vee  \\
& & (\shorttrans{(\omega\reductop_{V_{\Eta_\varepsilon}}, e :: \overline{e_\varepsilon},s_\varepsilon)}{\overline{m}}{\Delta_{\Theta_\varepsilon}}{(\omega'\reductop_{V_{\Eta_\varepsilon}}, \overline{e_\varepsilon} \lhd \overline{e}',s_\varepsilon')} \\ 
& &  \wedge  \emptyTrans{1})
\end{array}
\]
which can be simplified to

\[
\begin{array}{rll}
\multicolumn{2}{l}{\shorttrans{(\omega_1, e :: \overline{e_1}, (s_1,s_0))}{\overline{m}  \backslash E_{\Sigma_1}}{\Delta_{\Theta_1} || \Delta_{\Theta_\varepsilon}}{(\omega'_1, \overline{e_1}  \lhd ((\overline{e}' \cup \overline{m}) \cap E_{\Sigma_1}), (s_1',s_0))}} \\ 
\text{iff} \\
& \shorttrans{(\omega_1, e :: \overline{e_1},s_1)}{\overline{m}}{\Delta_{\Theta_1}}{(\omega'_1, \overline{e_1} \lhd \overline{e}',s_1')}, \\ 

\end{array}
\]

Now consider the resulting action relation \\
$\shorttrans{\omega}{a,\overline{m}}{\Omega_1 \pll \Omega_2}{\omega'}$ if
for some $i\in\{1,2\}$: $a
\in A_{\Eta_i}$ and $\shorttrans{\omega\reductop_{V_{\Eta_i}}}{a,\overline{m}}{\Omega_i}{\omega'\reductop_{V_{\Eta_i}}}$ and for $i
\neq j \in \{ 1, 2 \} : \omega\reductop_{V_{\Eta_j} \setminus V_{\Eta_i}} = \omega'\reductop_{V_{\Eta_j} \setminus V_{\Eta_i}}$ \\
which amounts to \\
$\shorttrans{\omega_1}{a,\overline{m}}{\Omega_1 \pll \Omega_2}{\omega'_1}$ if $a
\in A_{\Eta_1}$ and $\shorttrans{\omega_1}{a,\overline{m}}{\Omega_1}{\omega'_1}$ with $\omega_1,\omega'_1: V_{\Eta_1} \rightarrow \Val$ \\

It remains to show that the isomorphisms satisfy the coherence conditions
for symmetric monoidal categories. However, this follows easily since
the construction is based on set-theoretic union and product.
\end{proof}

\ \\
\textbf{Theorem \ref{th:InhDetIP}.}
\textit{Let $(\Omega_1, \Theta_1)$ and
$(\Omega_2, \Theta_2)$ be deterministic state machines and both action relations are compatible. Then $(\Omega_1 \pll \Omega_2, \Theta_1 \pll \Theta_2)$ is also deterministic.}

\begin{proof} \ \\
\textit{To show:} $\Omega_1 \pll \Omega_2$ is deterministic \\
If a shared action is triggered the action in the first machine must conform the changes of the same action in the second machine, i.e. shared actions trigger the same changes in shared variables.
Thus the action relation can be written as a action partial function 
\[
\Omega(\omega,a) = (\overline{m}_1 \cup \overline{m}_2,\omega^*)  
\]
iff 
$\Omega_1(\omega\reductop_{V_{\Eta_1}},a) = (\overline{m}_1,\omega'_1)$ and 
$\Omega_1(\omega\reductop_{V_{\Eta_2}},a) = (\overline{m}_2,\omega'_2)$
and
\[ \omega^*(v) = \begin{cases}
\omega'_1(v) & \text{if }  a \in A_{\Eta_1},  v \in V_{\Eta_1} \\ 
\omega(v) & \text{if }  a \in A_{\Eta_1}, v \in V_{\Eta_2}\setminus V_{\Eta_1}\\
\omega(v) & \text{if } a \in A_{\Eta_2}\setminus A_{\Eta_1}, v \in V_{\Eta_1}\setminus V_{\Eta_2}\\
\omega'_2(v) & \text{if } a \in A_{\Eta_2}\setminus A_{\Eta_1}, v \in V_{\Eta_2}\setminus,  V_{\Eta_1} \\ 

\end{cases}
\]
\\
\textit{To show:} $\Delta_{\Omega_1} \pll \Delta_{\Omega_2}$ is deterministic \\
Assume there are $\omega_1,\omega_1',\omega_1^*: V_{\Eta_1} \rightarrow \Val$, $ \omega_2,\omega_2',\omega_2^*: V_{\Eta_2} \rightarrow \Val$, $\overline{e} \subseteq (E_{\Sigma_1} \cup E_{\Sigma_2})$, $s_1,s_1',s_1^* \in S_{\Sigma_1}$, $s_2,s_2',s_2^* \in S_{\Sigma_2}$, $\overline{m} \subseteq M_{\Eta_2} \cup M_{\Eta_2}$ such that
\[
\shorttrans{(\omega,e :: \overline{e}, (s_1,s_2))}{\hspace{0cm}\overline{m}}{\hspace{0cm}\Delta_{\Theta_1} \pll \Delta_{\Theta_2}}{(\omega',\overline{e} \lhd (\overline{e}' \cup \overline{m}) \cap (E_{\Sigma_1} \cup E_{\Sigma_2}), (s_1',s_2'))}
\]
and
\[
\shorttrans{(\omega,e :: \overline{e}, (s_1,s_2))}{\hspace{0cm}\overline{m}}{\hspace{0cm}\Delta_{\Theta_1} \pll \Delta_{\Theta_2}}{(\omega^*,\overline{e} \lhd (\overline{e}^* \cup \overline{m}) \cap (E_{\Sigma_1} \cup E_{\Sigma_2}), (s_1^*,s_2^*))}
\]

Note that both transitions are triggered by the same event $e$. The event sets are disjoint. Thus both transitions must be caused by the same machine. Due to the commutativity of the interleaving product we can assume w.l.o.g. that it was the first machine. \\
Thus exist $\omega_1,\omega_1' \in \Omega_1$ such that $\omega\reductop_{V_{\Eta_1}} = \omega_1$, $\omega'\reductop_{V_{\Eta_1}} = \omega_1'$ and:
\[
\begin{array}{ll}
& \shorttrans{(\omega_1,e::\overline{e},s_1)}{\hspace{0cm}\overline{m}}{\hspace{0cm}\Delta_{\Theta_1}}{(\omega_1',\overline{e} \lhd (e' \cup \overline{m}) \cap E_{\Sigma_1},s_1')}  \\
& \shorttrans{(\omega_1,e::\overline{e},s_1)}{\hspace{0cm}\overline{m}}{\hspace{0cm}\Delta_{\Theta_1}}{(\omega_1^*,\overline{e} \lhd (e^* \cup \overline{m}) \cap E_{\Sigma_1},s_1^*)}
\end{array}
\]
The determinism of both machines yields \[(\omega_1',\overline{e} \lhd (e' \cup \overline{m}) \cap E_{\Sigma_1},s_1') = (\omega_1^*,\overline{e} \lhd (e^* \cup \overline{m}) \cap E_{\Sigma_1},s_1^*)\]
Furthermore, there can be no change in any variable, that is not an element of $V_{\Eta_1}$ and thus $\omega^*\reductop_{ V_{\Eta_2} \setminus V_{\Eta_1}} = \omega'\reductop_{ V_{\Eta_2}\setminus V_{\Eta_1}} =: \omega_R$ which yields $\omega^* = \omega^*_1 \cup \omega_R =\omega'_1 \cup \omega_R = \omega'$ and therefore
\[{(\omega_1',\overline{e} \lhd (e' \cup \overline{m}) \cap E_{\Sigma_1},s_1')}= {(\omega_1^*,\overline{e} \lhd (e^* \cup \overline{m}) \cap E_{\Sigma_1},s_1^*)}\]

The determinism of the messages sent in the transition relation follows from the determinism of the action relation.

\end{proof}

\textbf{Theorem \ref{thm:amalg-action}.}
\emph{The action institution admits weak amalgamation for pushout squares
with injective message mappings.}
\begin{proof}
First, we need some auxiliary definitions. 
One of them is the translation of a
model in the action institution along a signature morphism. Given
$\eta:\Sigma\to\Sigma'$ and a $\Sigma$-model $\Omega$,
let $\eta(\Omega)$ be defined by
$$  \xtrans{\omega_1}{\eta_A(a),\overline{m}}{\eta(\Omega)}{\omega_2}
\mbox{ iff }  \xtrans{\omega_1\reductop_{\eta_V}}{a,\eta_M^{-1}(\overline{m})}{\Omega}{\omega_2\reductop_{\eta_V}}. $$

We also need a modified version $\pll'$ of the interleaving product,
where messages of shared actions leading to compatible
states are united, instead of generating two separate transitions.\footnote
{Note that usually, for combination of state machines, one would use $\pll$.
However, $\pll'$ can be useful for expressing the semantics of
orthogonal regions in hierarchical state machines.}
  $\Omega_1 \pll' \Omega_2$ is given by

$$\shorttrans{\omega}{a,\overline{m}}{\Omega_1 \pll' \Omega_2}{\omega'}
\mbox{ if }\left\{
\begin{array}{l}
\exists i\in\{1,2\}\forall j \in \{ 1, 2 \}\setminus\{i\}\ . \\
\qquad a \in A_{\Eta_i}\!\setminus A_{H_j}\!\land
\shorttrans{\omega\reductop_{V_{\Eta_i}}}{a,\overline{m}}{\Omega_i}{\omega'\reductop_{V_{\Eta_i}}}\land
 \omega\reductop_{V_{\Eta_j}
  \setminus V_{\Eta_i}} = \omega'\reductop_{V_{\Eta_j} \setminus
  V_{\Eta_i}}\\
\lor \exists\overline{m_1},\overline{m_2}\ . \overline{m}=\overline{m_1}\cup\overline{m_2} \land\\
\qquad a \in A_{\Eta_1}\cap A_{\Eta_2}\land \forall i \in \{ 1, 2 \}\,.\,\shorttrans{\omega\reductop_{V_{\Eta_i}}}{a,\overline{m_i}}{\Omega_i}{\omega'\reductop_{V_{\Eta_i}}}
\end{array}\right.
$$

Now let a pushout with injective message mappings
\begin{equation*}
\begin{tikzpicture}[inner sep=0pt, outer sep=2pt]
  \matrix (m) [matrix of math nodes, ampersand replacement=\&, row sep=2.5em, column sep=2.5em, text height=1.7ex, text depth=0.25ex]{
    \Sigma   \& \Sigma_1\\
    \Sigma_2 \& \Sigma_R\\
  };
  \path[->,font=\scriptsize]
  (m-1-1) edge node [left] {$\sigma_2$} (m-2-1)
  (m-1-1) edge node [above] {$\sigma_1$} (m-1-2)
  (m-1-2) edge node [right] {$\theta_1$} (m-2-2)
  (m-2-1) edge node [below] {$\theta_2$} (m-2-2)
  ;
\end{tikzpicture}
\end{equation*}
be given, and assume that
$\Omega_1\reductop_{\sigma_1}=\Omega_2\reductop_{\sigma_2}$.

Since signature morphisms consist of mappings between sets,
it is easy to see that (surjection,injection)-factorisations exist.
Let $\theta_i$ be factorised as $\rho_i\circ\tau_i$. Then
the amalgamation is given by
$$\Omega_R=\tau_1(\Omega_1)\pll'\tau_2(\Omega_2),$$ 
The use of $\pll'$
together with injectivity of the message mappings ensures that
transitions in $\Omega_R$ reduce to transitions in the $\Omega_i$.
\end{proof}

\end{document}